\pdfoutput=1
\documentclass[a4paper,reqno]{amsart}

    \usepackage[english]{babel} %% English standard typographical conventions

    \usepackage{hyperref} %% hyperlinks
    \usepackage[capitalize]{cleveref} %% for arrays
    \usepackage{graphicx} %% figures, 
    \usepackage[]{color} %%more colors, monochrome=everything black
    \usepackage{csquotes} %% necessary for biber
    \usepackage{array} %% for arrays
    \usepackage{cancel}
    \usepackage{amsmath,amssymb,amsfonts,amsthm} %% ams packages
    \usepackage{mathtools,todonotes} %% add-on to amsmath
    \usepackage{tikz-cd}  %%For commutative diagrams with tikz
    \usepackage[all,cmtip]{xy} %% For diagrams
    \usepackage[bbgreekl]{mathbbol} %%For bold greek letters
\theoremstyle{definition} %%upright text, extra space above and below
    \newtheorem{definition}{Definition}

\theoremstyle{plain} %% italic text, extra space above and below
    \newtheorem{theorem}[definition]{Theorem}
    \newtheorem{proposition}[definition]{Proposition}
    \newtheorem{lemma}[definition]{Lemma}
    \newtheorem{corollary}[definition]{Corollary}

\theoremstyle{remark} %% upright text, no extra space above or below
    \newtheorem{remark}[definition]{Remark}

\usepackage[bibstyle=alphabetic,citestyle=alphabetic,useprefix,giveninits=true, minbibnames=99,maxbibnames=99,sorting=nyt, sortcites=true,backend=biber]{biblatex}
\renewbibmacro{in:}{} %% removes in: before journals

\DeclareSourcemap{
  \maps[datatype=bibtex]{
    \map[overwrite]{ %%If DOI is present, doesn't print arXiv
      \step[fieldsource=doi, final]
      \step[fieldset=url, null]
      \step[fieldset=eprint, null]
      }
     \map[overwrite]{ %% If only arXiv is present, doens't print pages and eid (eliminates duplicates)
      \step[fieldsource=eprint, final]
      \step[fieldset=pages, null]
      \step[fieldset=eid, null]
      \step[fieldset=journal, null]
    }  
  }
}

\DeclareMathOperator{\Ima}{Im}

\title[Scalar, SU(n), and Spinor Fields on Manifolds with Null-Boundary]{Gravity Coupled with Scalar, SU(n), and Spinor Fields on Manifolds with Null-Boundary}

\author{Alberto S. Cattaneo}
\address{Institut f\"ur Mathematik, Universit\"at Z\"urich, Winterthurerstrasse 190, 8057 Z\"urich, Switzerland}
\email{cattaneo@math.uzh.ch}

\author{Filippo Fila Robattino}
\address{Institut f\"ur Mathematik, Universit\"at Z\"urich, Winterthurerstrasse 190, 8057 Z\"urich, Switzerland}
\email{filippo.filarobattino@math.uzh.ch}

\author{Valentino Huang}
\address{}
\email{huangv@student.ethz.ch}

\author{Manuel Tecchiolli}
\address{Institut f\"ur Mathematik, Universit\"at Z\"urich, Winterthurerstrasse 190, 8057 Z\"urich, Switzerland}
\email{manuel.tecchiolli@math.uzh.ch}

\thanks{ASC, FFR, and MT acknowledge partial support of the SNF Grant No. 200020 192080. ASC acknowledges partial support of the Simons Collaboration on Global Categorical Symmetries. This research was (partly) supported by the NCCR SwissMAP, funded by the Swiss National Science Foundation. This article is based upon work from COST Action 21109 CaLISTA, supported   by COST (European Cooperation in Science and Technology) (www.cost.eu), MSCA-2021-SE-01-101086123 CaLIGOLA, and MSCA-DN CaLiForNIA-101119552. MT acknowledges support from the FK of the UZH. FFR acknowledges funding from the EU project Caligola HORIZON-MSCA-2021-SE-01, Project ID: 101086123.}

\begin{document}

\begin{abstract}
In this paper, we present a theory for gravity coupled with scalar, SU$(n)$ and spinor fields on manifolds with null-boundary. We perform the symplectic reduction of the space of boundary fields and give the constraints of the theory in terms of local functionals of boundary vielbein and connection. For the three different couplings, the analysis of the constraint algebra shows that the set of constraints does not form a first class system.

\end{abstract}

\maketitle
\tableofcontents

\section{Introduction}

The concept of gauge theories is a central aspect of modern mathematical physics, serving as the basis for the formulation of many fundamental physical theories. In a gauge theory, the physical fields are described by means of a principal $G$-bundle over a base manifold $M$ (possibly with boundary) and an action functional $S$ that embodies the symmetry of the theory and from which the field equations are derived. The conserved quantities of the theory come from the invariance of the action functional under the symmetry group, while interactions are introduced by gauging these symmetries, making them local. The mathematical representation of a gauge theory is achieved through a principal bundle $P$ and the gauge group can consequently be defined.

Two of the most widely accepted theories in fundamental physics are the Standard Model (SM) of particle physics and General Relativity (GR). The SM, with its symmetry group SU$(3)\times$SU$(2)\times$U(1), is a gauge theory that explains three of the four fundamental interactions. Meanwhile, GR is a theory that accounts for gravity, with physical fields representing the geometry of the base manifold.

GR is originally formulated using metric and Christoffel symbols, and is not formally equivalent to a Yang--Mills theory like the SM. In the SM, symmetries are encoded in physical fields through a principal connection. In order to study GR within a framework that would make the gauge formulation similar to the one of the SM, the coframe formalism, a.k.a. Palatini--Cartan (PC) theory, can be used. Within this formalism, the physical fields of GR are represented as coframes, a.k.a. vielbein, and local connections (see e.g. \cite{T2019}, \cite{thi2007} and \cite{C23}) and gravity becomes an explicit SO$(3,1)$ gauge theory.
 
In this paper, we examine the boundary structure of GR in the coframe formulation on manifolds with boundary, focusing specifically on the scenario where the boundary is null, resulting in a degenerate boundary metric. Our study extends the findings of the previous work \cite{CCR2022}, in which the authors investigated the geometric structures of gravity coupled to scalar, SU$(n)$, and spinor fields for a non-degenerate boundary metric. Our aim is to extend these findings to boundaries with the most general structure, namely including a possible degenerate metric. From a different perspective, this paper extends the study conducted in \cite{CCT21}, where the authors analyze the degenerate boundary structure of the Palatini--Cartan theory, in order to incorporate gravity coupled with matter and gauge fields. Therefore, the generalization of these results will lay the foundation for formulating the SM on manifolds with boundary (see \cite{SMB}).

The boundary structure is recovered performing the method established by Kijowski and Tulczijew (KT) in \cite{KT1979} (for an introduction see also \cite{C23} and references therein). This method involves characterizing the reduced phase space as a reduction, i.e., a quotient space, of the space of free boundary fields, rather than using the approach proposed by Dirac in \cite{Dirac1958}. The KT method has several mathematical advantages, including a more straightforward procedure for formulating constraints and compatibility with the BV-BFV formalism as described in \cite{CMR2012} (in the case of PC gravity this is done in \cite{CS2017b}, \cite{CCS2020} and \cite{CCS21b}). Additionally, the quantization procedure within the BV-BFV formalism as outlined in \cite{CMR2} can be more readily applied to the theory when using the KT approach.

Gravity in the coframe formalism is expressed though the so called Palatini--Cartan action. The structure of the symplectic form of the boundary fields poses a major challenge in the constraint analysis of the theory. This form is defined on a quotient space of the restrictions of the bulk fields to the boundary, determined by an equivalence relation given by the kernel of the map $e\wedge\,$, where $e$ denotes the coframe. To simplify the analysis, we describe this phase space using a fixed representative instead of working with equivalence classes, by introducing a suitable structural constraint. In prior works, such as \cite{CS2017}, \cite{CCS2020} and \cite{CCR2022}, this method has been successfully applied to space-like and time-like boundaries. However, for a null-boundary, the structural constraint must be adapted, as it only fixes the representative uniquely when the induced metric on the boundary is non-degenerate. We extend the solution proposed for space- and time-like boundaries to a null-boundary by adapting the structural constraint for all three different couplings (scalar, SU$(n)$, and spinor). The solution is slightly more involved and gives rise to second class constraints, compared to the non-degenerate case, where all constraints are first class.

\noindent
\textbf{Acknowledgments.} We thank Giovanni Canepa for useful comments and discussions.

\section{Geometrical background of gravity}

\subsection{Coframe formalism}

In the following section we will examine the geometrical background of the theory, i.e. the coframe formalism and the Palatini--Cartan action (see for example \cite{T2019, thi2007} and references therein).

The general set-up consists of
\begin{itemize}
    \item[--] An $N$-dimensional differentiable oriented\footnote{Orientability is not necessary (see, e.g., \cite[Section 2.1]{CCS2020}), but we assume it here for simplicity of notations.} pseudo-riemannian manifold $M$ with boundary $\Sigma$;
    \item[--] A principal GL$(N,\mathbb R)$-bundle $LM$ called the \emph{frame bundle}, which can be reduced to a principal SO$(N-1,1)$-bundle $P$;
    \item[--] An associated vector bundle $\mathcal V\coloneqq P\times_\rho V$ called the \emph{Minkowski bundle}, where $V$ is an $N$-dimensional real pseudo-riemannian vector space with reference metric $\eta=$diag$(1,...,-1)$ and $\rho\colon$SO$(N-1,1)\to$Aut$(V)$ is the fundamental representation of SO$(N-1,1)$.
\end{itemize}
Then, we define the vielbein via a reduction of the frame bundle.
\begin{definition}\label{coframedef}
	We define the \emph{vielbein} $\tilde e\colon P\to LM$ as the principal bundle isomorphism such that the following diagram commutes
\begin{center}
		\begin{tikzcd}
			P \ar[r,"\tilde e"]\ar[d,"p'"']  &LM \ar[d,"\pi'"] \\
			\mathcal{V} \ar[r,shift left,]
            &\ar[l,shift left,"e"] TM
		\end{tikzcd}
	\end{center}
 where $e\colon TM\to \mathcal V$ is the vector bundle isomorphism induced by $\tilde e\colon P\to LM$ and $p',\pi'$ the corresponding associated bundle maps. This means that the vielbein consists of the elements in $\Omega^1(M,\mathcal V)$ possessing smooth inverse. We can call this space $\tilde\Omega^1(M,\mathcal V)$.
\end{definition}

\begin{remark}~\\\vspace{-0.3cm}
\begin{itemize}
    \item[--] Given $i\colon\mathrm{SO}(N-1,1)\to\mathrm{GL}(N,\mathbb{R})$ as the canonical embedding, we recall that, in order for $\tilde e$ to be a principal bundle isomorphism, it must be an isomorphism of fiber bundles and also satisfy the equivariance condition
	\begin{align}
	   \mathcal{R}_{i(g)}\circ \tilde e=\tilde e\circ\mathcal{R}_g \quad \text{for all}\quad g\in G.
	\end{align}
    This is equivalent to asking that the following diagram commutes
        \begin{center}
		\begin{tikzcd}
			P \arrow[r," \tilde e"] \arrow[d,"\mathcal{R}_g"'] & LM \arrow[d,"\mathcal{R}_{i(g)}"] \\ P \arrow[r," \tilde e"'] & LM
		\end{tikzcd}
  	\end{center}
  \item[--] The existence and uniqueness of the map can be guaranteed through the use of the universal property of the quotient for the bundle isomorphism $\pi'\circ \,\tilde e\colon P\to TM$. This is possible thanks to the equivariance condition of $\tilde e$. The isomorphism property of the map $e\colon TM\to \mathcal V$ is simply inherited from $\tilde e$ by passing to the quotient.
    \item[--]Since the map $e\colon TM\to \mathcal V$ is an isomorphism of vector bundles, it acts like a linear isomorphism on the fibers. It means it can be written in the following way:
    \begin{align}
		{}_xe\colon T_xM & \to V \\
        v & \mapsto v^a=v^\mu e^a_\mu, \nonumber
    \end{align}
    where $v=v^\mu \partial_\mu\in T_xM$. Consider now the dual basis $\{dx^\mu\}$. We can collect the components of the isomorphism into the covector $e_\nu^a dx^\nu(\partial_\mu)=e_\mu^a$, since a covector is a linear map over the tangent space. Given that a basis of the cotangent space can be seen as a family of $N$ covectors $e^a_\mu dx^\mu$ and also that an isomorphism sends a basis to another basis, on a chart over $U\in M$, we can identify the map $e\colon TM\to \mathcal V$ with a family of $N$ covector fields or directly with a $V$-valued covector field in $\Omega^1(U,V)$. Therefore, if $M$ is parallelizable, we can identify the whole map $e$ with a $V$-valued covector field $e\in\Omega^1(M,V)$.
    \item[--] The name coframe formalism comes from the fact that $e$ not only defines an isomorphism, but, thanks to the fact that it is obtained from the reduction of the structure group of the frame bundle to the pseudo-orthogonal group SO$(N-1, 1)$, it is also a linear isometry on the fibers. In fact, the reduction to SO$(N-1, 1)$ means by definition that the frames of the frame bundle are orthonormal, namely we have on the fibers $g_{\mu\nu}e_a^\mu e_b^\nu=\eta_{ab}$. On the other hand, in terms of their dual basis (coframes) $\{e^a\}$, we have $g_{\mu\nu}=\eta_{ab}e^a_\mu e^b_\nu$, which can be written as
        \begin{align}\label{isometrycondition}
		g=e^*\eta.
	\end{align}
    This means that $e$ is a linear isometry.
    \end{itemize}
\end{remark}

    \begin{proposition}\label{prop:isowedge}
        The inner product on $V$ allows the identification $\mathfrak{so}(N-1,1) \cong \textstyle{\bigwedge^2} {V}$.
    \end{proposition}

Because of this proposition, we can identify $\mathfrak{so}(N-1,1)$-valued forms\footnote{In the sense of a vector bundle with fibers $\mathfrak{so}(N-1,1)$} with $\textstyle{\bigwedge^2} \mathcal V$-valued forms and we will use the following shortened notation to indicate the spaces of $i$-forms on $M$ with values in the $j$th wedge product of $\mathcal{V}$
\begin{align}\label{e:NotationOmega}
    \Omega^{i,j}\coloneqq \Omega^i\left(M, \textstyle{\bigwedge^j} \mathcal{V}\right),
\end{align}
which is generalized to all possible $i,j\in \mathbb N$.\\
\begin{remark}
    These spaces form indeed a graded algebra with graded product
\begin{align*}
    \wedge \colon \Omega^{i,j} \times   \Omega^{k,l}  &\rightarrow  \Omega^{i+k,j+l} & & \text{for} \; i+k \leq N, \; j+l \leq N \\
    (\alpha,  \beta)  &\mapsto  \alpha \wedge \beta=(-1)^{(i+j)(k+l)}\beta \wedge \alpha. & & 
\end{align*}
We will refer to an alement in $\Omega^{i,j}$ also as an $(i,j)$-form.
\end{remark}
\begin{definition}
    A connection form $\omega$ on a principal $G$-bundle $P$ is a $\mathfrak g$-valued $1$-form on $P$ such that:
\begin{itemize}
    \item[--] It is adjoint-equivariant;
    \item[--] For each $\xi\in\mathfrak g$ and fundamental vector field $X_\xi$, it holds $\omega(X_\xi)=\xi$.
\end{itemize}
\end{definition}
We will refer to the space of principal connections on P as $\mathcal A(P)$.
\begin{remark}
  If we consider a principal connection form on the principal SO$(N-1,1)$-bundle $P$, namely an element $\omega\in\Omega^1(P,\textstyle{\bigwedge^2} V)$ (thanks to \cref{prop:isowedge}), we can pull it back using local sections. We will obtain a family of local connections $\omega_\alpha\in\Omega^1(U_\alpha,\textstyle{\bigwedge^2}\mathcal V)$. These forms define a covariant derivative on $M$ (see \cref{covder}).
\end{remark}

The action of the Lie algebra on $\textstyle{\bigwedge^j}\mathcal V$-differential forms will be denoted in the following way:

\begin{definition}\label{generalizedlie}
    Let $\alpha\in\Omega^{i,j}$ and $\beta\in\Omega^{k,l}$. Then, we define the generalized Lie bracket
    
    \begin{align*}
    [\, \,, \, ]\colon \Omega^{i,j}\times\Omega^{k,l}\to\Omega^{i+k,j+l-2}
    \end{align*}
    through
    \begin{equation}
    \begin{split}
    &[\alpha,\beta]_{\mu_1\dots \mu_{i+k}}^{a_1\dots a_{j+l-2}}=\\[7pt]
    &=\sum_{\sigma_{i+k}}\sum_{\sigma_{j+l-2}} \text{sign}(\sigma_{i+k})\text{sign}(\sigma_{j+l-2}) \alpha_{\mu_{\sigma(1)} \dots \mu_{\sigma(i)}}^{a_{\sigma(1)}\dots a_{\sigma(j-1)}a} \beta_{\mu_{\sigma(i+1)} \dots \mu_{\sigma(i+k)}}^{a_{\sigma(j)}\dots a_{\sigma(j+l-2)}b}\iota(\rho)_{ab},
    \end{split}
    \end{equation}
    where $\iota(\rho)$ is the contraction map that $\textstyle{\bigwedge^m}\mathcal V$ inherits\footnote{The representation $\rho$ induces an algebra representation $\mathrm{d}\rho$ and we can translate that to $\bigwedge^2\mathcal V$ thanks to \cref{prop:isowedge}. Then, we can easily generalize this action to $\textstyle{\bigwedge^m}\mathcal V$.} from the representation $\rho$ of SO$(N-1,1)$. For the fundamental representation, this map is just the contraction with the $\eta$.
\end{definition}
\begin{remark}
    Shortly speaking, the brackets act as a wedge product on both space-time and internal indices not contracted with the contraction map.
\end{remark}
\begin{remark}
    The contraction map $\iota(\rho)$ is obtained from the representation map of the algebra $\mathrm{d}\rho\colon\mathfrak so(N-1,1)\to\text{End}(V)$ composed with the isomorphism of \cref{prop:isowedge}.
\end{remark}
\begin{definition}\label{covder}
    Local connections define an exterior covariant derivative for $\textstyle{\bigwedge^j}\mathcal V$-valued
$i$-forms on $M$. We denote such a map with
\begin{align}
    d_\omega\colon \Omega^{i,j}\to \Omega^{i+1,j}.
\end{align}
Explicitly, it reads
\begin{align}
    d_\omega\alpha=d\alpha+[\omega, \alpha],
\end{align}
where $\alpha\in\Omega^{i,j}$.
\end{definition}
\begin{remark}
    Note that the representation of the brackets is the fundamental one. This is due to the fact that $\mathcal V$ is the associated bundle to $P$ through the fundamental representation. In the case of a different associated bundle, through a different representation, the brackets will be replaced by the given representation.
\end{remark}
\begin{definition}
    Let $\omega\in\mathcal A(P)$ be a principal connection. Then, the associated local connections on $M$ define a global $2$-form $F_\omega\in\Omega^{2,2}$, which satisfies, in any arbitrary trivialization chart $(U_\alpha, s_\alpha)$,
    \begin{align}
        F_\omega|_{U_\alpha}=d\omega_\alpha+\frac12[\omega_\alpha, \omega_\alpha],
    \end{align}
    with $\omega_\alpha=s_\alpha^*\omega$.
\end{definition}
A more detailed derivation of this definition can be found in \cite{T2019}.\\
\begin{definition}\label{PC_def}
    The classical Palatini--Cartan theory is the assignment of the pair $(\mathcal F_{PC}, \mathcal S_{PC})_M$ to every pseudo riemannian $N$-dimensional manifold and vector space $V$ with reference metric\footnote{Note that any particular choice of the Lorentzian structure on $V$ is immaterial, since a change in $V$ would just isomorphically reflect to the space of fields without changing $\mathcal S_{PC}$.} $\eta$ with space of fields
    \begin{align}
    \mathcal F_{PC}=\tilde\Omega^{1,1}\times \mathcal A(P)\ni (e,\omega)
    \end{align}
    and action functional
    \begin{align}\label{eq:actionPCT}
    \mathcal S_{PC}=\int_{M}\frac1{(N-2)!}e^{N-2} F_\omega+\frac1{N!}\Lambda e^N,
    \end{align}
    where $\Lambda \in\mathbb R$ and the powers in $e$ are in terms of the wedge product.
\end{definition}
\begin{remark}
    In \cref{eq:actionPCT}, we have omitted both the wedge product and the trace. The trace operator is the map $\mathrm{Tr}\colon\textstyle{\bigwedge^N} V\to\mathbb R$ such that, given a basis $\{u_i\}_{i=1,...,N}$ of $V$, it holds that $\mathrm{Tr}[u_{i_1}\wedge \cdot\cdot\cdot\wedge u_{i_n}]=\varepsilon_{i_1...i_n}$, thus the trace works as a choice of the orientation on $M$, which must be compatible with the SO$(N-1,1)$ reduction.
\end{remark}
\begin{remark}  
    In the subsequent sections, we will avoid reiterating a similar definition for each distinct case. Rather, we will provide the space of fields on $M$, and it is important to bear in mind that the definitions of the upcoming theories will be straightforward generalizations of \cref{PC_def}.
\end{remark}
The Euler--Lagrange equations coming from the action principle $\delta \mathcal S_{PC}=0$ are, respectively for the variations in $e$ and $\omega$,
Euler--Lagrange equations of Palatini--Cartan theory read:
\label{eq:el}
\begin{align}
\frac1{(N-3)!}e^{N-3}F_\omega-\frac1{(N-1)!}\Lambda e^{N-1}&=0\label{eq:PCe}\\[7pt]
e^{N-3}d_\omega e&= 0,\label{eq:PComega}
\end{align}
which, in $N=4$, reduce to
\label{eq:el4}
\begin{align}
e F_\omega-\frac1{3!}\Lambda e^{3}&=0\label{eq:PCe4}\\[7pt]
e d_\omega e&= 0.\label{eq:PComega4}
\end{align}
By injectivity of the map $e\wedge\cdot$ on $(2,1)$-forms, \cref{eq:PComega} is equivalent to
\begin{equation}\label{e:tf}
d_\omega e = 0,
\end{equation}
which is the torsion-free condition. Therefore, this is the equation that identifies the Levi-Civita connection for the metric \eqref{isometrycondition}.

\subsection{Symplectic reduction on the boundary}

The geometrical method implemented to the study of the boundary structure of the theory is the KT construction described in \cite{KT1979}.\\
The construction starts from a space of bulk fields denoted with $\mathcal F$ and an action functional of such fields denoted with $\mathcal S$. In the case of the Palatini--Cartan theory, these are precisely $\mathcal F_{PC}$ and $\mathcal S_{PC}$. We notice that the integration by parts in the variation of the action $\mathcal S$ gives rise to a boundary term
\begin{align}\label{noetherform}
    \alpha=\int_{\Sigma}\frac1{(N-2)!}e^{N-2}\delta \omega,
\end{align}
which we call the Noether $1$-form.

By considering the pull-back of the fields in $\mathcal F$ to the boundary $\Sigma$ via the natural inclusion $i\colon\Sigma\to M$, we obtain the space of pulled-back fields denoted by $\mathcal{\tilde{F}}_\Sigma$. In this setting, the boundary term $\alpha$ defined in \eqref{noetherform} can be interpreted as a $1$-form on the space of pulled-back fields. Furthermore, the variational operator $\delta$ is regarded as a de Rham differential of the complex of differential forms on $\mathcal{\tilde{F}}_\Sigma$.

Note that the $2$-form on $\mathcal{\tilde{F}}_\Sigma$ defined via
    \begin{align}\label{degsympform}
        \tilde\varpi=\delta\alpha=\int_{\Sigma}\frac1{(N-3)!}e^{N-3}\delta e \delta \omega
    \end{align}
is closed.
\begin{remark}\label{sympformdeg}
    It is important to note that a closed $2$-form does not necessarily have to be non-degenerate. The form $\tilde\varpi$ defined in \cref{degsympform} may have a non-trivial kernel. This is the case with both the free theory and the all three different couplings analyzed in this paper. A closed $2$-form with possibly degenerate kernel is called a pre-symplectic form and a space endowed with such a form is called a pre-symplectic space.
\end{remark}
\begin{definition}\label{PC_preboundary}
    We define the \emph{space of pre-boundary fields} for the Palatini--Cartan theory as the pre-symplectic space $(\mathcal{\tilde{F}}_\Sigma,\tilde\varpi)$, where $\mathcal{\tilde{F}}_\Sigma$ is the space of pulled-back fields on the boundary $\Sigma$ and $\tilde\varpi$ is the pre-symplectic form.
\end{definition}
\begin{remark}
    We will use the following definition for bundle valued differential forms on the boundary
    \begin{align}
        \Omega^{i,j}_\Sigma\coloneqq\Omega^i(\Sigma,\textstyle{\bigwedge^j}i^*\mathcal V).
    \end{align}
\end{remark}

As we pointed out in \cref{sympformdeg}, the pre-symplectic form might be indeed not symplectic (it might be degenerate). In order to obtain a symplectic space, we could just quotient by the distribution given by the kernel of the pre-symplectic form.

\begin{definition}\label{def:geometricphasespace}
	We define the \emph{geometric phase space} of the theory as the symplectic space $(\mathcal{{F}}_\Sigma,\varpi)$ obtained as the quotient of the space of pre-boundary fields by the kernel of its pre-symplectic form\footnote{This quotient is to be intended in the sense of distributions on the tangent bundle. Note that $\mathrm{ker}(\tilde\varpi)$ is involutive, since $\tilde\varpi$ is closed.}
	\begin{equation}
		\mathcal{{F}}_\Sigma \coloneqq \frac{\mathcal{\tilde{F}}_\Sigma}{\mathrm{ker}(\tilde{\varpi})}
	\end{equation}
 and with symplectic form $\varpi$, the unique $2$-form on $\mathcal{{F}}_\Sigma$ such that $p^*\varpi = \tilde\varpi$, where $p\colon\mathcal{\tilde F}_\Sigma\to \mathcal{F}_\Sigma$ is the canonical projection.
\end{definition}

In field theory, it is commonly understood that not all field equations are dynamical, and on a manifold with boundary, this is equivalent to having some field equations that are non-transverse with respect to the boundary. The resulting non-dynamical equations can be interpreted as constraints that must be satisfied by the boundary fields.\\
We can give these constraints the form of local functionals on $\mathcal{{F}}_\Sigma$ (or $\mathcal{\tilde{F}}_\Sigma$), just by integrating the pulled-back equations on $\Sigma$. We denote this set of constraints as $\mathcal C$ (or $\mathcal{\tilde C}$).

The first understanding of the nature of a set of constraints on a symplectic space is due to Dirac \cite{Dir1950}. He pointed out correctly that the nature of the constraints, which he divided in first- and second-class, had important implications on the local degrees of freedom of the theory\footnote{The local degrees of freedom are defined as the dimension of the reduced phase space and the dimension of a space is define as the rank of the fiber or its dimension as a $C^\infty$-module.}. More precisely, the hamiltonian vector fields of the first-class constraints generate the algebra of the symmetry group of the theory and the ones of the second-class constraints do not.  In \cref{sectio_firstsecond}, a more detailed discussion of first- and second-class constraints is presented.

The vanishing locus of these integral constraints, quotiented by the action of the algebra generated by their hamiltonian vector fields, is called the \emph{reduced phase space}. Roughly speaking, this is the space of the non-gauge equivalent (thanks to the quotient) initial conditions (the fields are on the boundary) for the dynamical field equations of the theory (since we have considered the vanishing locus of the constraints).

In Table 1, we summarize all the steps to the reduced phase space\footnote{This table is taken from \cite{CCT21}}.

\begin{table}[ht]
\begin{tikzcd}[row sep=0.75cm] 
(\mathcal{F}, S )
\arrow[d,dashed, "\text{Pull-back to the boundary}"]\\
(\mathcal{\tilde{F}}_\Sigma, \tilde\varpi, \tilde{\mathcal{C}})
\arrow[d, "\text{Reduction by ker$(\tilde{\varpi})$}"]\\
(\mathcal{{F}}_\Sigma, {\varpi}, \mathcal{C})
\arrow[d, "\text{Vanishing locus + /$\sim$ gauge algebra}"]\\
\text{Reduced Phase Space}
\end{tikzcd}
\vspace{0.25cm}
\caption{}
\vspace{-1cm}
\end{table}

\subsection{The structural and degeneracy constraints}\label{sect:strconstr}

From now on, we will work in $N=4$.\\
On the boundary $\Sigma$, the injectivity property of the map $e\wedge\cdot$ acting on boundary $(2,1)$-forms is lost.\footnote{See \cite{CS2017}.} This property guaranteed the equivalence of $d_\omega e=0$ and $e d_\omega e=0$ in the bulk. This situation is indeed problematic. In fact, in the bulk we have two perfectly equivalent conditions, namely two equivalent ways of writing one of the field equations. When we pull these back to the boundary, we want this equivalence to hold in order to make sense of the field equations on the boundary themselves. In other words, since in the bulk $e d_\omega e=0$ must give rise to the same solution space of $d_\omega e=0$, if the solutions space of these two equations on the boundary were to differ, then the Cauchy problem would be ill-defined. I.e., the pull-back to the boundary of the solutions obtained from the field equations in the bulk would be different from the boundary fields obtained from the solutions of the fields equations on the boundary. It means that one has to impose some additional conditions in order to maintain this equivalence on the boundary. We call part of the family of these extra conditions the \emph{structural constraint}.
    
This problem is present in both the non-degenerate and degenerate cases; however, the form of the structural constraint strictly depends on the nature of the boundary (null or non-null). In fact, in the non-degenerate case, the structural constraint alone is sufficient to ensure the aforementioned equivalence on the boundary. On the other hand, on a null boundary, the extra conditions split into a structural and a degeneracy constraint. We will see that, from a  different perspective, the structural constraint of the non-degenerate case is just a specific characterization of the structural and the degeneracy constraint where the latter is trivial.

Note that the core of this section, as we will mention again later in \cref{remark_wlog1}, is maintained rather general, namely independent of the field equations. The application of these results to the Palatini--Cartan theory is, on the one hand, a fundamental building block for the subsequent sections and, on the other hand, a useful way to get a solid grasp on the ideas behind the main results of the section itself.

First, we start by giving some definitions.

\begin{definition}
Let $e\in \Omega_\Sigma^{1,1}$ and $e^k\in\Omega_\Sigma^{k,k}$ be the wedge product of $k$ elements $e$. Then, we define the following maps:
    \begin{align}
                W_k^{\Sigma, (i,j)}\colon \Omega_\Sigma^{i,j}  & \longrightarrow \Omega_\Sigma^{i+k,j+k} \\
                \alpha  & \longmapsto    e^k\wedge \alpha \nonumber
    \end{align}  
    \begin{align}
                    \varrho^{(i,j)} \colon \Omega_{\Sigma}^{i,j}  & \longrightarrow \Omega_{\Sigma}^{i+1,j-1} \\
                    \alpha & \longmapsto [e,\alpha] \nonumber
    \end{align}
    \begin{align}\label{maptilderho}
                    \tilde{\varrho}^{(i,j)} \colon \Omega_{\Sigma}^{i,j}  & \longrightarrow \Omega_{\Sigma}^{i+1,j-1} \\
                    \alpha & \longmapsto [\tilde e,\alpha], \nonumber
    \end{align}
with $\tilde e\in \tilde\Omega_\Sigma^{1,1}$ being a degenerate vielbein, namely $\tilde e^*\eta=0$.
\end{definition}

We also give the definitions of three geometrical objects that we will require in the following theorems.

\begin{definition}\label{fund_spaces}
Let $J$ be a complement\footnote{To obtain an explicit expression for the complement, one can follow these steps. Start by selecting an arbitrary Riemannian metric on the boundary manifold $\Sigma$ and extend it to the space $\Omega^{2,1}$. Then, the orthogonal complement of the image of the map $\varrho^{(1,2)} |{\mathrm{Ker} W_{1}^{\Sigma, (1,2)}}$ in $\Omega_{\Sigma}^{2,1}$ can be identified as the space $J$, with respect to the chosen Riemannian metric.} in $\Omega_{\Sigma}^{2,1}$ of the space $\Ima \varrho^{(1,2)} |_{\mathrm{Ker} W_{1}^{\Sigma, (1,2)}}$. Then, we define the following subspaces:
    \begin{align}
    \mathcal{T}&\coloneqq \mathrm{Ker}W_{1}^{\Sigma (2,1)} \cap J \subset \Omega_{\Sigma}^{2,1}\\
    \mathcal S&\coloneqq\mathrm{Ker}W_{1}^{\Sigma, (1,3)} \cap \mathrm{Ker} \tilde{\varrho}^{(1,3)}  \subset \Omega^{1,3}_{\Sigma}\\
    \mathcal K&\coloneqq \mathrm{Ker}W_{1}^{\Sigma, (1,2)} \cap \mathrm{Ker} \varrho^{(1,2)} \subset \Omega_{\Sigma}^{1,2}.
\end{align}
\end{definition}
We present the initial key result for the degenerate theory, which will ensure the equivalence between $d_\omega e=0$ and $e d_\omega e=0$ at the boundary. While it may appear initially quite redundant with respect to \cref{strconstr_free}, it will have profound implications for the geometry of the theory, as highlighted in \cref{geom_impl}.
\begin{lemma}[Corollary of \cref{strconstr_free}]\label{strconstr_free_deg}
Let $e_n\in\Omega_\Sigma^{0,1}$ be fixed such that, for a chosen vielbein $e\in\tilde\Omega_\Sigma^{1,1}$, $\{e(v_1),e(v_2),e(v_3),e_n\}$\footnote{Notice in particular that, in any neighborhood of $e$ of the space of boundary fields, we are allowed to pick $e_n$ independently of the dynamics of the vielbein $e$. In other words, we can state that $e_n$ is constant in the field $e$. This trivially implies that $e_n$ has no variation along $e$.} is a basis of $i^*\mathcal V$, where $\{v_1,v_2,v_3\}$ is a basis of $T\Sigma$. Moreover, let $\alpha\in\Omega_\Sigma^{2,1}$. Then, we have that
$$\alpha = 0$$
if and only if
	\begin{equation}\label{tau_strcontr}
        \begin{cases}  \alpha\in\mathrm{Ker}W_1^{\Sigma, (2,1)}\\[6pt]
        e_n (\alpha- p_{\mathcal{T}}\alpha) \in \Ima W_{1}^{\Sigma,(1,1)}\\[6pt]
        p_{\mathcal T}\alpha=0,
        \end{cases}
    \end{equation}
where $p_{\mathcal T}$ is the projector onto $\mathcal T$. We call the second and third conditions in \eqref{tau_strcontr} respectively the structural and the degeneracy constraints.
\end{lemma}
The next lemma provides a formulation of the degeneracy constraint in terms of an integral functional.
\begin{lemma}\label{lemma_tau_1}
    Let $\alpha\in\Omega_\Sigma^{2,1}$. Then, we have the following equivalence
    	\begin{equation}
        p_{\mathcal T}\alpha=0\quad\Longleftrightarrow\quad\int_\Sigma \tau\alpha=0\quad\forall\tau\in\mathcal{S}.
	   \end{equation}
\end{lemma}
\begin{proof}
    See \cite{CCT21}.
\end{proof}
\begin{remark}\label{remark_wlog1}
    As long as we do not specify any $\alpha$, these two lemmas remain purely geometrical and do not depend on the properties of the field equations. We will then be able to use these results for the interactive theories where the equivalence condition on the boundary will differ from $d_\omega e=0$ and $e d_\omega e=0$ (since the field equations will be different themselves). Therefore, in general, we need to specify $\alpha$ for each different theory. In particular, for the Palatini--Cartan theory, $\alpha=d_\omega e$ and the structural and the degeneracy constraints read
    \begin{equation}\label{tau_strcontr_PC}
        \begin{cases}  e_n (d_\omega e- p_{\mathcal{T}}d_\omega e) \in \Ima W_{1}^{\Sigma,(1,1)}\\[6pt]
        p_{\mathcal T}d_\omega e=0.
        \end{cases}
    \end{equation}
\end{remark}
\begin{remark}\label{geom_impl}
    It is important to emphasize that Eq.s \eqref{tau_strcontr} are trivially equivalent to the structural constraint
    \begin{align}\label{str_constr_nd}
        e_n \alpha\in \Ima W_{1}^{\Sigma,(1,1)}
    \end{align}
    in the non-degenerate case. Nonetheless, the introduction of this split plays a crucial role in the analysis of the degenerate theory. More specifically, apart from $p_{\mathcal T}$ not being trivial, \cref{str_constr_nd} alone will not be sufficient to uniquely fix a representative of the equivalence class defining the symplectic space (see \cref{mastertheorem1}). In other words, since in the non-degenerate case $p_{\mathcal T}\alpha=0$ holds trivially, we can infer that the second equation in \eqref{tau_strcontr} is the most general form of the structural constraint of the theory, whose geometrical implications are only visible in the degenerate case. In fact, the peculiar integral condition of the degenerate case, introduced in \cref{lemma_tau_1}, carries significant consequences. It can be interpreted as a modification of the set of constraints of the theory by incorporating a new functional constraint. For $\alpha=d_\omega e$ (the case of the Palatini--Cartan theory), this is denoted as
    	\begin{equation}\label{constrtau}
        R_\tau=\int_\Sigma \tau d_\omega e.
	\end{equation}
 Further discussions of this matter will be presented in the next section.
\end{remark}

\subsection{Fixing the representative}\label{sect:fixrepr}
The reduction by the kernel of the presymplectic form, as shown in \cite{CS2017}, is equivalent to a quotient space with an equivalence relation on the connection form, as stated in the following theorem.

\begin{theorem}\label{mastertheorem1}
The geometric phase space for the Palatini--Cartan theory is the symplectic manifold ($\mathcal F_{\Sigma}, \varpi)$ given by the following equivalence relation on the space of pre-boundary fields $\mathcal{\tilde{F}}_\Sigma$
\begin{align}\label{eqcl}
    \omega^\prime\sim\omega\quad\Longleftrightarrow\quad\omega^\prime-\omega \in\mathrm{Ker}W_{1}^{\Sigma, (1,2)}
\end{align}
and the symplectic form
    \begin{align}
        \varpi=\int_{\Sigma}e\delta e \delta [\omega].
    \end{align}
We refer to this equivalence class as $\mathcal A(i^* P)_{red}$.
\end{theorem}

\begin{proof}
    See \cite{CS2017}.
\end{proof}

\begin{remark}
    To study the reduced phase space of the theory, we make use of representatives for the equivalence classes defined in \eqref{eqcl}. In the non-degenerate case, these representatives are uniquely determined by the structural constraint itself. In other words, ensuring the equivalence of $d_\omega e=0$ and $e d_\omega e=0$ on the boundary, is enough to determine uniquely the representatives of the equivalence classes defined in \eqref{eqcl}. However, in the degenerate case, the structural constraint and the degenerate constraint (or its integral form $R_\tau$), despite the fact that they indeed ensure on the boundary the equivalence mentioned above, are not sufficient to uniquely assign a representative to each equivalence class. Therefore, it is necessary to seek an alternative way to guarantee the unambiguous determination of these representatives.
 \end{remark}
 We can accomplish that through the following lemma.
\begin{lemma}\label{uniquerep_free}
Let $i^*g$ be degenerate. Then, given $\omega\in\Omega_\Sigma^{1,2}$ and $e_n\in\Omega_\Sigma^{0,1}$ as in \cref{strconstr_free_deg}, the conditions
    \begin{equation}
        \begin{cases}\label{eq:uniquerep_free}
             e_n (d_\omega e- p_{\mathcal{T}}(d_{\omega} e)) \in \Ima W_{1}^{\Sigma,(1,1)}\\[6pt]
           p_{\mathcal{K}} \omega = 0
        \end{cases}
    \end{equation}
    uniquely define a representative of the equivalence class $[\omega]\in\mathcal A(i^* P)_{red}$.
\end{lemma}

\begin{proof}
    See \cite{CCT21}.
\end{proof}

\begin{remark}
    In \cite{CCT21}, it has been proved that the analysis is independent of the choice of the representative of the equivalence class \eqref{eqcl}. In more rigorous terms, for each choice of the representatives there is a canonical symplectomorphism between the symplectic space defined by representatives and the geometric phase space of the theory.
\end{remark}

\begin{remark}
    It is important to highlight that, in the non-degenerate case, the subspaces $\mathcal{T}$, $\mathcal{S}$, and $\mathcal{K}$ of \cref{fund_spaces} are trivial. It follows that the projectors $p_\mathcal{K}$ and $p_\mathcal{T}$ are also trivial. Once again, this means that, in the non-degenerate theory, the structural constraint alone serves the purpose of establishing the equivalence between $d_\omega e=0$ and $e d_\omega e=0$ on the boundary, as well as uniquely determining the representatives of the equivalence classes defined in \cref{eqcl}.
\end{remark}

We have seen that, on a null-boundary, we need both the structural and the degeneracy constraints together with the additional equation $p_\mathcal{K}\omega=0$ in order to both guarantee the equivalence between $d_\omega e=0$ and $e d_\omega e=0$ on the boundary and uniquely fix the representative of the equivalence class $[\omega]\in\mathcal A(i^* P)_{red}$.\\
More specifically, the role of the structural constraint together with the integral constraint $R_\tau$ is the one of ensuring the aforementioned equivalence condition, whereas, the structural constraint together with $p_\mathcal{K}\omega=0$ will uniquely fix the representatives.

We display now the constraints of the theory.
\begin{definition}\label{constraints_gravity-deg}
Let\footnote{The notation $[1]$ indicates a shift in parity.} $c\in\Omega_\Sigma^{0,2}[1]$, $\xi\in\mathfrak{X}(\Sigma)[1]$, $\lambda\in C^\infty(\Sigma)[1]$ and $\tau\in\mathcal{S}[1]$. Then, we define the following functionals
    \begin{align}
		L_c & = \int_\Sigma ced_\omega e \\[3pt]
		P_\xi & = \int_\Sigma \frac{1}{2}\iota_\xi(e^2)F_\omega+\iota_\xi(\omega-\omega_0)ed_\omega e \\[3pt]
		H_\lambda & = \int_\Sigma \lambda e_n\Big(eF_\omega+\frac{\Lambda}{3!}e^3\Big)\\[3pt]
		R_\tau & = \int_\Sigma \tau d_\omega e.
	\end{align}
We refer to these as the constraints of the Palatini--Cartan (degenerate) theory.
\end{definition}
We are now able to determine the algebra of the constraints of the theory. This differs from the one of the non-degenerate theory, since the new constraint $R_\tau$ changes the nature of the Poisson brackets, which become second class.
\begin{theorem}\label{thm:Brackets_constraints}
    Let $i^*g$ be degenerate. Then the structure of the Poisson brackets of the constraints $L_c$, $P_{\xi}$, $H_{\lambda}$ and $R_{\tau}$ is given by the following expressions
    \begin{align*}
    &\begin{aligned}
        	&\{L_c, L_c\}  = - \frac{1}{2} L_{[c,c]} & \qquad\qquad
       	& \{P_{\xi}, P_{\xi}\}  =  \frac{1}{2}P_{[\xi, \xi]}- \frac{1}{2}L_{\iota_{\xi}\iota_{\xi}F_{\omega_0}} \\
       	& \{L_c, P_{\xi}\}  =  L_{\mathcal{L}_{\xi}^{\omega_0}c} &
       	& \{H_{\lambda},H_{\lambda}\}  \approx 0 \\
       	&  \{L_c, R_{\tau}\} =  -R_{p_{\mathcal{S}}[c, \tau]} & 
       	& \{P_{\xi},R_{\tau}\}  = R_{p_{\mathcal{S}}\mathcal{L}_{\xi}^{\omega_0}\tau}.\\
       	&\{ R_{\tau}, H_{\lambda} \}   \approx G_{\lambda \tau}  &
 		& \{R_{\tau},R_{\tau}\}  \approx F_{\tau \tau}
        \end{aligned}\\
	&\{L_c,  H_{\lambda}\}
        = - P_{X^{(a)}} + L_{X^{(a)}(\omega - \omega_0)_a} - H_{X^{(n)}}\\        
        &\{P_{\xi},H_{\lambda}\} =  P_{Y^{(a)}} -L_{ Y^{(a)} (\omega - \omega_0)_a} +H_{ Y^{(n)}} 
    \end{align*}
    with $X=[c,\lambda e_n]$ and $Y=\mathcal{L}_\xi^{\omega_0}(\lambda e_n)$ and where the superscripts $(a)$ and $(n)$ describe their components with respect to $e_a,e_n$. Furthermore $F_{\tau \tau}$ and $G_{\lambda \tau}$ are functionals of $e$, $\omega$, $\tau $ and $\lambda$ that are not proportional to any other constraint.
\end{theorem}
\begin{remark}
	The symbol $\approx$ indicates the identity on the zero locus of the constraints. In particular, this means that those brackets written with this symbol are not a linear combination of the constraints themselves. On the other hand, all the brackets written with a = vanish on the zero locus, for example $\{L_c,L_c\}\approx0$.
\end{remark}
For the details of the proof and the definition of first and second class constraints, we refer to \cite{CCT21}.
\begin{remark}
    The former results still hold in the case where the boundary metric has some extra degeneracy, apart from the one coming from the nature of the null-boundary. In the upcoming sections, we will, for the sake of simplicity, assume that the restriction of the metric, excluding the degenerate direction of the boundary, remains non-degenerate.
\end{remark}

\begin{remark}
    The distinctive feature of the degenerate theory, highlighted in \cite{CCT21} and summarized in \cref{sectio_firstsecond}, is that the additional constraint $R_\tau$ turn out to be second-class (see \cref{def_firstsecond}). As discussed in \cref{sectio_firstsecond}, this implies that the requirement to uniquely determine a representative for $[\omega]$ results in the reduced phase space of the theory having a dimension of two, compared to the non-degenerate theory, which had a dimension of four.
\end{remark}

The first step in each of the following sections will be the one of finding the correct set of equations as a choice for the structural constraint. Then, we will find the relations for uniquely fixing the representatives. Once established the correct geometrical set-up, we will proceed by computing the algebra of the constraints of the theory at hand. This will be done in the cases of scalar, SU$(n)$ and spinor couplings.

\section{Coupling terms: degenerate structure}
\subsection{Scalar field}\label{sect:scalar}
In the following section, we will dive into the case of the scalar coupling. We stress that we refer to \cite{CCR2022} for what concerns the non-degenerate case.\\
In the canonical formalism, the Palatini--Cartan theory coupled with a scalar field maintains the very same geometrical background of the previous sections with the addition of two new fields, the scalar field $\phi$ and its conjugate momentum (upon equation of motion) $\Pi$.

We must define the building blocks of our scalar Palatini--Cartan theory, starting with the space of fields on $M$, which reads
    \begin{align}
    \mathcal F^{\phi}=\tilde\Omega^{1,1}\times \mathcal A(P)\times C^\infty(M)\times\Omega^{0,1}\ni (e,\omega,\phi,\Pi),
    \end{align}
    and the action functional
    \begin{align}\label{eq:actionPCT_scalar}
    \mathcal S^{\phi}=\mathcal S_{PC}+\int_{M}\frac{1}{6}e^{3}\Pi d\phi+\frac{1}{48}e^4(\Pi,\Pi),
    \end{align}
where the brackets indicates the inner product of the Minkowski bundle. It follows that the Euler-Lagrange equations of the theory are given by
    \begin{align}
	ed_\omega e&=0\\\label{EL_domegae_scalar}
	eF_\omega+\frac{\Lambda}{6}e^{3}+\frac{1}{2}e^{2}\Pi d\phi+\frac{1}{12}e^{3}(\Pi,\Pi)&=0  \\
	d(e^{3}\Pi) &= 0\\
	e^{3}(d\phi-(e,\Pi)) &= 0.
    \end{align}
The variation of the action leads to the following Noether $1$-form on the space of pre-boundary fields\footnote{Note that we keep the same notation for the Noether and the pre-symplectic forms of the previous section even though these are different.}
\begin{align}
    	\tilde{\alpha} = \int_\Sigma \frac{1}{2}e^{2}\delta\omega+\frac{1}{6}e^{3}\Pi\delta\phi,
\end{align}
which gives rise to the following pre-symplectic form
\begin{align}\label{eq:scal_presympl_form}
		\tilde{\varpi} =\delta\tilde{\alpha} = \int_\Sigma e\delta e\delta\omega+\frac{1}{6}\delta(e^3\Pi)\delta\phi.
\end{align}
Similarly to the previous section, we can define the space of pre-boundary fields $\tilde{\mathcal F}^{\phi}_\Sigma$, as in \cref{PC_preboundary} for the Palatini--Cartan theory, by pulling back the fields to the boundary $\Sigma$. Also in this case, we will write the fields on the boundary with the same letters as for those in the bulk.

As shown in \cite{CCR2022}, we are now able to define the geometric phase space of the theory via a reduction through the kernel of the pre-symplectic form.\\
Here is a generalization of \cref{mastertheorem1}.
\begin{theorem}\label{mastertheorem_scalar}
The geometric phase space for the scalar Palatini--Cartan theory is the symplectic manifold $(\mathcal F^{\phi}_\Sigma,\varpi)$ given by the following equivalence relations on the space of pre-boundary fields $\tilde{\mathcal F}^{\phi}_\Sigma$
\begin{align}\label{eqcl_scalar}
    \omega^\prime\sim\omega\quad&\Longleftrightarrow\quad\omega^\prime-\omega \in\mathrm{Ker}W_{1}^{\Sigma, (1,2)}\\[4pt]
    \Pi^\prime\sim\Pi\quad&\Longleftrightarrow\quad\Pi^\prime-\Pi \in\mathrm{Ker}W_{3}^{\Sigma, (0,1)}
\end{align}
and the symplectic form
    \begin{align}
        \varpi=\int_{\Sigma}e\delta e \delta [\omega]+\frac{1}{6}\delta(e^3[\Pi])\delta\phi.
    \end{align}
We refer to these equivalence classes as $\mathcal A(i^* P)_{red}$ and $\Omega_{\Sigma, red}^{0,1}$.
\end{theorem}
\begin{proof}
    See \cite{CCR2022}.
\end{proof}
We notice that the first field equation \eqref{EL_domegae_scalar} does not couple with the scalar field. Therefore, since this purely geometrical term is equivalent to the one of the Palatini--Cartan theory (namely $\alpha=d_\omega e$), the structural and degeneracy constraints possess the same form of the free theory. In fact, as we said, they serve the purpose of maintaining the equivalence between $e d_\omega e=0$ and $d_\omega e=0$ on the boundary. We recall here the aforementioned constraints, which thus read
    \begin{equation}\label{tau_strcontr_scalar}
        \begin{cases}  e_n (d_\omega e- p_{\mathcal{T}}d_\omega e) \in \Ima W_{1}^{\Sigma,(1,1)}\\[6pt]
        p_{\mathcal T}d_\omega e=0.
        \end{cases}
    \end{equation}
Similarly to the Palatini--Cartan theory, we focus on fixing the representative of the equivalence classes defined in \cref{mastertheorem_scalar}. The purely gravitational part remains the same, since it follows uniquely from the kernel of the piece of \eqref{eq:scal_presympl_form} equal to the pre-symplectic form of the free Palatini--Cartan theory. In other words, since in the present case the equivalence class $[\omega]$ is defined in the same way of the Palatini--Cartan theory, as well as the structural constraint, it follows that \cref{uniquerep_free} applies verbatim to the scalar field theory.\\
Although, we are left to fix the representative of the equivalence class for $\Pi$. For this purpose, we give the following lemma.
\begin{lemma}\label{uniquerep_scalar}
Let $i^*g$ be degenerate. Then, given $\phi\in C^\infty(\Sigma)$, $\Pi\in\Omega_{\Sigma}^{0,1}$ and $e_n\in\Omega_\Sigma^{0,1}$ as in \cref{strconstr_free_deg}, the conditions
    \begin{equation}
        \begin{cases}\label{eq:uniquerep_scalar}
             d\phi-(e,\Pi)=0\\[6pt]
           p_{W}\Pi=0,
        \end{cases}
    \end{equation}
with\footnote{Here, we  regard the boundary metric as a map $i^*g\colon T\Sigma\to T^*\Sigma$ and therefore we have that $\mathrm{Ker}(i^*g)=\{\xi\in\mathfrak{X}(\Sigma)\mid\iota_\xi (i^*g)=0\}\subset\mathfrak{X}(\Sigma)$.} ${W}=e(\mathrm{Ker}(i^*g))$, uniquely define a representative of the equivalence class $[\Pi]\in\Omega_{\Sigma, red}^{0,1}$.    
\end{lemma}
\begin{proof}
We first notice that, if we consider a vector field along the degenerate direction, namely $X\in\mathrm{Ker}(i^*g)$, and we take the contraction of the field equations with it, we obtain the condition
\begin{align}\label{select_orth_phi}
    \iota_X d\phi=0.
\end{align}
What happens is that the degeneracy in the boundary metric decouples $\phi$ and $\Pi$ along the degenerate direction\footnote{This gives a condition on the derivative of $\phi$. More specifically, the degeneracy of the boundary metric complicates the selection of the components of the fields in the orthogonal direction to the boundary. This implies that we could potentially have some spurious components of the field $\phi$ generating the diffeomorphisms along the orthogonal direction. Therefore, we can interpret the condition of \cref{select_orth_phi} as a geometrical constraint that selects the only component of the symmetry transformations orthogonal to the boundary which are actually generated by the Hamiltonian vector field $h_\lambda^\phi$ of \cref{h_phi}. In other words, one could say that \cref{select_orth_phi} selects the “physically meaningful" components of the derivative of the field $\phi$.} and this is precisely why we need, compared to the non-degenerate case, an extra condition in order to fix the representative of the equivalence class of $\Pi$.

We can decompose the field $\Pi\in\Omega_\Sigma^{0,1}$ in the following way\footnote{We take the basis of $i^*\mathcal V$ given by the vielbein and the completion $e_n$. Notice that, as a section of $i^*\mathcal V$, $e_n$ will have components along the vielbein in general.}
\begin{align}\label{dec_Pi}
    \Pi=\pi^ne_n+\pi^ae_a,
\end{align}
with $a=1,2,3$. Then, we notice that, thanks to definition of the wedge product, $e^3e_a=0$ for every $a$ and therefore $\pi=\pi^ae_a\in\mathrm{Ker}W_3^{\Sigma,(0,1)}$. This means that fixing a certain $\pi^n$ uniquely defines an equivalence class $[\Pi]\in\Omega_{\Sigma, red}^{0,1}$ and vice versa. We are thus left to show that the conditions \eqref{eq:uniquerep_scalar} fix also uniquely $\pi=\pi^a e_a$, as a function of $\pi^n$. Now, we recall that $\mathrm{dim}(\mathrm{Ker}(i^*g))=1$ and $e$ is injective and, therefore, we have that $W\subset e(T\Sigma)$ is a 1-dimensional subspace. Furthermore, for any open neighbourhood of $e(T\Sigma)$, without loss of generality, we can assume that the basis given by the vielbein $\{e_1,e_2,e_3\}$ is such that, say, $e_3$ spans $W$. From \cref{dec_Pi}, it follows that the condition
\begin{align}
    p_W\Pi=0
\end{align}
implies
\begin{align}\label{P_WPi=0} 
    \pi^ne^3_n+\pi^3=0.
\end{align}
Moreover, with such a choice of basis, we can write the exterior derivative of the scalar field as
\begin{align}\label{iota_X dphi=0}
    d\phi = \partial_i\phi dx^i = e_1^a\partial_a\phi dx^1+e_2^a\partial_a\phi dx^2,
\end{align}
where we implemented the condition $\iota_X d\phi=0$, which reads $e^a_3\partial_a\phi=0$. Lastly, we can write the field equations implementing \cref{iota_X dphi=0}, obtaining
    \begin{align}
		d\phi-(e,\Pi) & = \partial_i\phi dx^i-(e^a_idx^ie_a,\pi^be_b+\pi^ne_n) \\
		& = e_i^a\partial_a\phi dx^i-e_i^a\pi^bg_{ab}dx^i-e_i^a\pi^ng_{an}dx^i \\
		& = e_i^a(\partial_a\phi-\pi^bg_{ab}-\pi^ng_{an})dx^i = 0,
	\end{align}
where $g$ is the metric and $b=1,2$. Since the restricted inner product is non-degenerate, we have
\begin{align}
    \partial_a\phi-\pi^bg_{ab}-\pi^ng_{an}=0
\end{align}
and thus we deduce the following equation for $\pi^b$ (with $b=1,2$)
\begin{align}\label{fixpi12}
    \pi^b=g^{ab}(\partial_a\phi-\pi^n g_{an}).
\end{align}
It follows that \cref{P_WPi=0,fixpi12} completely fix the components of $\pi$ in terms of $\pi^n$. Hence, since fixing $\pi^n$ is equivalent to fixing a representative for $[\Pi]$ and vice versa, we have that, given an equivalence class (or equivalently a $\pi^n$), the conditions \eqref{eq:uniquerep_scalar} fix uniquely the representative of $[\Pi]$. On the other hand, given a representative, the conditions \eqref{eq:uniquerep_scalar} fix unambiguously a $\pi^n$ and therefore an equivalence class $[\Pi]$.
\end{proof}
We have uniquely determined the representatives for the equivalence classes that define the symplectic space of boundary fields. As a result, we can now write the set of constraints of the theory as functionals of the representatives themselves.
\begin{definition}\label{constraints_SPC}
Let $c\in\Omega_\Sigma^{0,2}[1]$, $\xi\in\mathfrak{X}(\Sigma)[1]$, $\lambda\in C^\infty(\Sigma)[1]$ and $\tau\in\mathcal{S}[1]$. Then, we define the following functionals
    \begin{align}
		L_c & = \int_\Sigma ced_\omega e \\[3pt]
		P^{\phi}_\xi & = \int_\Sigma \frac{1}{2}\iota_\xi(e^2)F_\omega+\frac{1}{3!}\iota_\xi(e^3\Pi)d\phi+\iota_\xi(\omega-\omega_0)ed_\omega e \\[5pt]
		H^{\phi}_\lambda & = \int_\Sigma \lambda e_n\Big(eF_\omega+\frac{\Lambda}{3!}e^3+\frac{1}{2}e^2\Pi d\phi+\frac{1}{2\cdot3!}e^3(\Pi,\Pi)\Big)\\
		R_\tau & = \int_\Sigma \tau d_\omega e.
	\end{align}
We refer to these as the constraints of the scalar Palatini--Cartan theory.
\end{definition}
In the following theorem, we give the form of the Poisson brackets determining the constraint algebra of the theory.
\begin{theorem}\label{poiss_brack_scalar}
    Let $i^*g$ be degenerate. Then the Poisson brackets of the constraints of \cref{constraints_SPC} read
        \begin{align*}
    &\begin{aligned}
        	&\{L_c, L_c\}  = - \frac{1}{2} L_{[c,c]} & \qquad\qquad
       	& \{P^{\phi}_{\xi}, P^{\phi}_{\xi}\}  =  \frac{1}{2}P^{\phi}_{[\xi, \xi]}- \frac{1}{2}L_{\iota_{\xi}\iota_{\xi}F_{\omega_0}} \\
       	& \{L_c, P^{\phi}_{\xi}\}  =  L_{\mathcal{L}_{\xi}^{\omega_0}c} &
       	& \{H^{\phi}_{\lambda},H^{\phi}_{\lambda}\}  \approx 0 \\
       	&  \{L_c, R_{\tau}\} =  -R_{p_{\mathcal{S}}[c, \tau]} & 
       	& \{P^{\phi}_{\xi},R_{\tau}\}  = R_{p_{\mathcal{S}}\mathcal{L}_{\xi}^{\omega_0}\tau}\\
       	&\{ R_{\tau}, H^{\phi}_{\lambda} \}   \approx G_{\lambda \tau}&
 		& \{R_{\tau},R_{\tau}\}  \approx F_{\tau \tau}
        \end{aligned}\\
	&\{L_c,  H^{\phi}_{\lambda}\}
        = - P^{\phi}_{X^{(a)}} + L_{X^{(a)}(\omega - \omega_0)_a} - H^{\phi}_{X^{(n)}}\\        
        &\{P^{\phi}_{\xi},H^{\phi}_{\lambda}\} =  P^{\phi}_{Y^{(a)}} -L_{ Y^{(a)} (\omega - \omega_0)_a} +H^{\phi}_{ Y^{(n)}},
    \end{align*}
	with $X=[c,\lambda e_n]$ and $Y=\mathcal{L}_\xi^{\omega_0}(\lambda e_n)$ and where the superscripts $(a)$ and $(n)$ describe their components with respect to $e_a,e_n$. Furthermore, $F_{\tau\tau}$ and $G_{\lambda\tau}$ are functionals of $e,\omega,\tau$ and $\lambda$ which are not proportional to any other constraint.\footnote{They are properly defined in \cite{CCT21} (proof of Theorem $30$).}
\end{theorem}
\begin{proof}
    First, we introduce the following notation\footnote{With $P_\xi$ and $H_\lambda$ of \cref{constraints_gravity-deg}.}
    \begin{align}\label{split_int_scalar}
        &P^{\phi}_\xi = P_\xi+p^{\phi}_\xi& \,\, &H^{\phi}_\lambda =H_\lambda+h^{\phi}_\lambda,
    \end{align}
in order to simplify the computations.

In accordance with the results from \cite{CCT21} and \cite{CCR2022}, we possess knowledge of the some of the brackets as follows
\begin{align*}
    &\begin{aligned}
        	&\{L_c,L_c\} = -\frac{1}{2}L_{[c,c]} & \qquad\qquad
       	& \{L_c,P^{\phi}_\xi\}   = L_{\mathcal{L}_\xi^{\omega_0}c} \\[2pt]
            &\{P^{\phi}_\xi,P^{\phi}_\xi\} = \frac{1}{2}P^{\phi}_{[\xi,\xi]}-\frac{1}{2}L_{\iota_\xi\iota_\xi F_{\omega_0}}& \qquad\qquad
            & \{L_c,R_\tau\} = -R_{p_\mathcal{S}[c,\tau]}\\[2pt]
            &\{P_\xi,R_\tau\}=R_{p_\mathcal{S}\mathcal{L}_\xi^{\omega_0}\tau}& \qquad\qquad
            &\{ R_{\tau}, H_{\lambda} \}   \approx G_{\lambda \tau} \\[2pt]
            &\{R_{\tau},R_{\tau}\}  \approx F_{\tau \tau}& \qquad\qquad
            &\{H^\phi_\lambda,H_\lambda^{\phi}\}  \approx 0\\[2pt]
        \end{aligned}\\
	   &\{L_c,H^{\phi}_\lambda\}=-P^{\phi}_{X^{(a)}}+L_{X^{(a)}(\omega-\omega_0)_a}-H^\phi_{X^{(n)}}\\[2pt] 
        &\{P_\xi^{\phi},H^{\phi}_\lambda\}=P^{\phi}_{Y^{(a)}}-L_{Y^{(a)}(\omega-\omega_0)_a}+H^{\phi}_{Y^{(n)}}
    \end{align*}
with $F$ and $G$ non-identically-vanishing functional of $\tau$ and $\lambda$ defined in \cite{CCT21} (Theorem $30$) and $X=[c,\lambda e_n]\in\Omega_\Sigma^{0,1}$ divided into a tangential component $X^{(a)}=[c,\lambda e_n]^{(a)}$ and a normal component $X^{(n)}=[c,\lambda e_n]^{(n)}$. We are therefore left to compute the brackets $\{R_\tau,h^{\phi}_\lambda\}$ and $\{R_\tau,p_\xi^{\phi}\}$. Also, we can recall the known results from \cite{CCT21} and \cite{CCT21} for what concerns all Hamiltonian vector fields. In particular, for $p_\xi^{\phi}$, we have
	\begin{align}
		\mathbb{p}_e^{\phi} & = 0 & \mathbb{p}_\omega^{\phi} & = 0 \\
		\mathbb{p}_\rho^{\phi} & = -\mathcal{L}_\xi^{\omega_0}\rho & \mathbb{p}_\phi^{\phi} & = -\xi(\phi),
	\end{align}
whereas, for $h^{\phi}_\lambda$, we have
	\begin{align}
		\mathbb{h}^{\phi}_e & = 0 & \mathbb{h}^{\phi}_\omega& = \lambda e_n\big(\Pi d\phi+\frac{e}{4}(\Pi,\Pi)\big)-\frac{\lambda}{2}e\Pi(\Pi,e_n) \\
		\mathbb{h}^{\phi}_\rho & = \frac{1}{2}d_\omega(\lambda e_ne^2\Pi) & \mathbb{h}^{\phi}_\phi & = -\lambda(e_n,\Pi)\label{h_phi},
	\end{align}
 where we have defined a new field $\rho\coloneqq\frac{1}{3!}e^3\Pi\in\Omega_\Sigma^{3,4}$.
 
Next, it is helpful to write explicitly the variation\footnote{We compute the Hamiltonian vector fields in the following manner. Let $\mathbb X$ be the Hamiltonian vector field of the functional $F$ for the symplectic form $\varpi$, then it holds $\iota_{\mathbb X}\varpi-\delta F=0$, where $\delta F$ is the functional derivative of $F$. } of $R_\tau$, which reads\footnote{Since $\tau$ is defined on $\mathcal S$ and the latter is defined making use of $e$, it follows that $\tau$ has a non-trivial variation along $e$.}
	\begin{align}
		\delta R_\tau & = \int_\Sigma \delta\tau d_\omega e-\tau[\delta\omega,e]+\tau d_\omega\delta e \\
		& = \int_\Sigma (g(\tau,\omega,e)+d_\omega\tau)\delta e+[\tau,e]\delta\omega,
	\end{align}
	where we have introduced the formal expression $g=g(\tau,e,\omega)$ which encodes the dependence of $\tau$ on $e$ (see \cite{CCT21} Theorem $30$ for further details). It follows that the Hamiltonian vector fields are
	\begin{align}
		e\mathbb{R}_e & = [\tau,e] & e\mathbb{R}_\omega & = g(\tau,\omega,e)+d_\omega\tau \\
		\mathbb{R}_\rho & = 0 & \mathbb{R}_\phi & = 0.
	\end{align}
Now that we possess all Hamiltonian vector fields, we are ready to compute the Poisson brackets of the remaining constraints. First, we notice that
\begin{align}
    \{R_\tau, p_\xi^{\phi}\}=0
\end{align}
since $p_\xi^{\phi}$ has trivial Hamiltonian vector fields along $e$ or $\omega$. Then, we compute
\begin{align}
    \{R_\tau, h_\lambda^{\phi}\} = \int_\Sigma -\lambda e_n\Pi d\phi[\tau,e]-\frac{e}{4}\lambda e_n(\Pi,\Pi)[\tau,e]+\frac{\lambda}{2}e\Pi(\Pi,e_n)[\tau,e].
\end{align}
Here, the last two terms are zero thanks to $e[\tau,e]=0$ (following from $e\tau=0$ in the definition of $\mathcal S$). For the term bracket, we have
	\begin{align}
		\int_\Sigma -\lambda e_n\Pi d\phi[\tau,e] & = \int_\Sigma -\lambda\Pi d\phi\tau[e_n,e]\\[4pt]
        &=\int_\Sigma -\lambda\Pi d\phi\tau(e_n,e) \\[4pt]
        &\hspace{-0.12cm}\overset{\eqref{eq:uniquerep_scalar}}{=}\int_\Sigma -\lambda\Pi (e,\Pi)\tau(e_n,e) \\[4pt]
        &\hspace{-0.12cm}=\int_\Sigma -\lambda\Pi (e,\Pi)e_n[\tau, e] \\[4pt]
        &=0,
	\end{align}
where we implemented the Leibniz identity for the squared brackets, the definition of $\tau\in\mathcal S$ and \cref{diagtetr}, thanks to the fact that there are no derivatives in the integral\footnote{Roughly speaking, we can “diagonalize" the vielbein.}.

Finally, in order to complete the proof, we can simply exploit the linearity of the Poisson brackets and recall the definition of the split introduced in \cref{split_int_scalar} together with the known results mentioned above.
\end{proof}

\subsection{Yang--Mills field}
In this section, we will examine the case of an SU$(n)$-gauge-field\footnote{All the considerations below work with a general Lie algebra $\mathfrak g$.}, namely a principal connection $A$ of a principal SU$(n)$-bundle over $M$ denoted with $P$ (see \cite{T2019} Section $5$). It follows that the space of gauge fields is locally modelled on $\Omega^1(M,\mathfrak{su}(n))$, via the pull-backs along the sections of $G$. In the Standard Model of particle physics, this kind of field is responsible for the mediation of a variety of interactions, in particular, the Electroweak and the Strong interaction. Moreover, similarly to what we did in the previous section, we associate to the gauge field\footnote{Note that we refer to both $A$ and its pull-back as the gauge field.} $A$ an independent field $B\in\Gamma\big(\textstyle{\bigwedge^2}\mathcal{V}\otimes \mathfrak{su}(n)\big)$.

Hence, the Yang--Mills--Palatini--Cartan theory is defined by the following space of fields
\begin{align}
    \mathcal F^{A}=\tilde\Omega^{1,1}\times \mathcal A(P)\times \mathcal A(G)\times\Gamma\big(\textstyle{\bigwedge^2}\mathcal{V}\otimes \mathfrak{su}(n)\big)\ni (e,\omega,A,B),
    \end{align}
    and the action functional
    \begin{align}\label{eq:actionPCYM}
    \mathcal S^{A}=\mathcal S_{PC}+\int_M \frac{1}{4}e^{2}\mathrm{Tr}(BF_A)+\frac{1}{46!}e^4\mathrm{Tr}(B,B),
    \end{align}
where $\Omega^2(M,\mathfrak{su}(n))\ni F_A=d A+\frac12[A,A]$ is the field strength, $(\,\,,\,)$ is the canonical pairing in $\bigwedge^2\mathcal V$ and $\mathrm{Tr}\colon\mathfrak{su}(n)\to\mathbb R$ is the trace over the algebra.

The Euler-Lagrange equations are as follows
    \begin{align}
	d_\omega e&=0\\
	e(F_\omega+\mathrm{Tr}(BF_A))+\frac{e^{3}}{6}(\Lambda+\frac{1}{2}\mathrm{Tr}(B,B))&=0\\
	e^{2}\big(F_A+\frac{1}{2}(e^2,B)\big)&=0\\
	d_A(e^{2}B)&=0,
    \end{align}
whereas the Noether $1$-form becomes
\begin{align}
    \tilde{\alpha} = \int_{\Sigma}\frac{e^{2}}{2}\delta\omega+\frac{e^{2}}{2}\mathrm{Tr}(B\delta A).
\end{align}
It follows that the pre-symplectic form of the theory is
\begin{align}
    \tilde{\varpi} = \delta\tilde{\alpha} = \int_\Sigma e\delta e\delta\omega+\mathrm{Tr}(eB\delta e\delta A)+\frac{1}{2}\mathrm{Tr}(e^2\delta B\delta A).
\end{align}
    This is a $2$-form over the space of pre-boundary fields obtained as the pull-back of bulk fields along $i\colon\Sigma \to M$ and denoted in this case as $\tilde{\mathcal{F}}^{A}_\Sigma$. Notice that, also in this case, we refer to boundary fields with the same notation of bulk fields.
\begin{theorem}\label{mastertheorem_YM}
The geometric phase space for the Yang--Mills--Palatini--Cartan theory is the symplectic manifold $(\mathcal F^{A}_\Sigma,\varpi)$ given by the following equivalence relations on the space of pre-boundary fields $\tilde{\mathcal F}^{A}_\Sigma$
\begin{align}\label{eqcl_YM}
    \omega^\prime\sim\omega\quad&\Longleftrightarrow\quad\omega^\prime-\omega \in\mathrm{Ker}W_{1}^{\Sigma, (1,2)}\\[4pt]
    B^\prime\sim B\quad&\Longleftrightarrow\quad B^\prime-B \in\mathrm{Ker}W_{2}^{\Sigma, (0,2*)},
\end{align}
where $2*$ indicates that the $\bigwedge^2\mathcal V$-algebra is tensored with $\mathfrak{su}(n)$, and the symplectic form
    \begin{align}
       \varpi = \int_\Sigma e\delta e\delta[\omega] + \frac{1}{2}\mathrm{Tr}(\delta(e^2[B])\delta A).
    \end{align}
We refer to these equivalence classes as $\mathcal A(i^* P)_{red}$ and $\Gamma\big(\textstyle{\bigwedge^2}i^*\mathcal{V}\otimes \mathfrak{su}(n)\big)_{red}$.
\end{theorem}
\begin{proof}
    See \cite{CCR2022}.
\end{proof}
\begin{remark}
In the context of the Yang--Mills--Palatini--Cartan theory, we can indeed establish unique representatives for these equivalence classes. Subsequently, we can proceed to formulate the constraints in a manner analogous to the approach we previously employed in the preceding section. The representative for $[\omega]\in\mathcal A(i^* P)_{red}$ is already uniquely fixed thanks to equivalent considerations to the ones articulated in the previous sections. Therefore, the structural and the degeneracy constraints for the Yang--Mills--Palatini--Cartan theory read
     \begin{equation}\label{tau_strcontr_YM}
        \begin{cases}  e_n (d_\omega e- p_{\mathcal{T}}d_\omega e) \in \Ima W_{1}^{\Sigma,(1,1)}\\[6pt]
        p_{\mathcal T}d_\omega e=0.
        \end{cases}
    \end{equation}
\end{remark}
We are therefore left with the problem of the representative for $[B]\in\Gamma\big(\textstyle{\bigwedge^2}i^*\mathcal{V}\otimes \mathfrak{su}(n)\big)_{red}$, which is determined by the following lemma.
\begin{lemma}\label{uniquerep_YM}
Let $i^*g$ be degenerate. Then, given $A\in\mathcal A(i^*G)$, $B\in\Gamma\big(\textstyle{\bigwedge^2}i^*\mathcal{V}\otimes \mathfrak{su}(n)\big)$ and $e_n\in\Omega_\Sigma^{0,1}$ as in \cref{strconstr_free_deg}, the conditions
    \begin{equation}
        \begin{cases}\label{eq:uniquerep_YM}
            F_A+\frac{1}{2}(e^2,B)=0\\[6pt]
            p_{\Omega^{0,1*}_{e}\wedge W}B=0,
        \end{cases}
    \end{equation}
with $\Omega^{i,j*}_{e}\coloneqq\Omega^i\big(\Sigma, \bigwedge^je(T\Sigma)\otimes \mathfrak{su}(n)\big)$ where $W=e(\mathrm{Ker}(i^*g))$, uniquely define a representative of the equivalence class $[B]\in\Gamma\big(\textstyle{\bigwedge^2}i^*\mathcal{V}\otimes \mathfrak{su}(n)\big)_{red}$.    
\end{lemma}
\begin{proof}
We can decompose an element\footnote{We can consider the basis for $e(T\Sigma)$ given by the vielbein. See the proof of \cref{uniquerep_scalar} for more details.} $B\in\Gamma\big(\textstyle{\bigwedge^2}i^*\mathcal{V}\otimes \mathfrak{su}(n)\big)$ as\footnote{Apart from the wedge product, in order to lighten the notation, we also omit the tensor product.} 
\begin{align}
    B=b^{an}e_ae_n+\frac12 b^{ab}e_ae_b,
\end{align}
with $b^{an},b^{ab}\in\Gamma(\mathfrak{su}(n))$ and $a,b=1,2,3$. We notice that $b=b^{ab} e_ae_b\in\mathrm{Ker}(W_2^{\Sigma,(0,2*)})$, since $e^2e_ae_b=0$ for all $a,b=1,2,3$. This directly implies that the components $b^{an}$ are already uniquely determined by the equivalence class $[B]$ and vice versa. Now, as we did in the proof of \cref{uniquerep_scalar}, we observe that $\mathrm{dim}(\mathrm{Ker}(i^*g))=1$ and $e$ is injective and, therefore, we have that $W\subset e(T\Sigma)$ is a 1-dimensional subspace. Hence, for any open neighbourhood of $e(T\Sigma)$, without loss of generality, we can take as a basis of $e(T\Sigma)$ the one given by the $\{e_1,e_2,e_3\}$ such that $e_3$ spans $W$. Then, since a basis of $\Omega^{0,1*}_{e}\wedge W$ is given by $\{e_1e_3, e_2e_3\}\otimes\mathfrak{su}(n)$, we have that, similarly to the scalar case, we first notice that the field equations imply the condition $\iota_X F_A=0$, with $X\in\mathrm{Ker}(i^*g)$. Moreover, the condition $p_{\Omega^{0,1*}_{e}\wedge W}B=0$ implies that
\begin{align}\label{PWF=0}
    2b^{[1n}e^{3]}_n+b^{13}=2b^{[2n}e^{3]}_n+b^{23}=0,
\end{align}
where the square brackets in the indices denote the anti-symmetrization.

Next, consider the condition $\iota_X F_A=0$. Then, we can write
	\begin{align}
		F_A = \frac12 F_{ab}e^a_i e^b_j dx^i dx^j=\frac12F_{12}dx^1dx^2.
	\end{align}
Furthermore, similarly to the preceding case, we can write
	\begin{align}
		&2F_A+(e^2,B)=\\
        &=F_{ij}dx^idx^j+(\frac12e_i^ae_j^bdx^idx^je_ae_b,b^{cd}e_ce_d+b^{cn}e_ce_n) \\[2pt]
        &=F_{ab}e^a_ie^b_jdx^idx^j+b^{cd}e_i^ae_j^bg_{ac}g_{bd}dx^idx^j+b^{cn}e_i^ae_j^bg_{ac}g_{bn}dx^idx^j \\[2.5pt]
		&=e_i^ae_j^b(F_{ab}+b^{cd}g_{ac}g_{bd}+b^{cn}g_{ac}g_{bn})dx^idx^j = 0.
	\end{align}
We observe that, since the restricted inner product is non-degenerate, we have
	\begin{align}
		F_{ab}+b^{cd}g_{ac}g_{bd}+b^{cn}g_{ac}g_{bn}=0 
    \end{align}
and, given $a,b,c,d\neq3$, we can use the inverse metric to write
    \begin{align}
        b^{cd} = -(g^{ac}g^{bd}F_{ab}+g^{bd}g_{bn}b^{cn}).
	\end{align}
This result together with \cref{PWF=0} fixes uniquely the elements $b^{ab}$ in terms of $b^{an}$ (with $a,b=1,2,3$). The completion of the proof follows from analogous considerations to the ones of the scalar case in the previous section.
\end{proof}
We are now able to give the definition of the constraints of the theory.
\begin{definition}\label{constraints_YMPC}
Let $c\in\Omega_\Sigma^{0,2}[1]$, $\mu\in C^\infty(\Sigma,\mathfrak{g})[1]$, $\xi\in\mathfrak{X}(\Sigma)[1]$, $\lambda\in C^\infty(\Sigma)[1]$ and $\tau\in\mathcal{S}[1]$. Moreover, let $\rho= e^2B\in\Omega_\Sigma^{2,4*}$. Then, we define the following functionals
    \begin{align}
		L_c & = \int_\Sigma ced_\omega e \\[3pt]
        M_\mu &= \int_\Sigma \mathrm{Tr}(\mu d_A\rho) \\[3pt]
		P^{A}_\xi & = \int_\Sigma \frac{1}{2}\iota_\xi e^2F_\omega+\iota_\xi(\omega-\omega_0)ed_\omega e+\frac{1}{2}\mathrm{Tr}(\iota_\xi\rho F_A)\\
        & \phantom{= \int_\Sigma}+\mathrm{Tr}\big(\iota_\xi(A-A_0)d_A\rho\big)\\
		H^{A}_\lambda & = \int_\Sigma \lambda e_n\Big(eF_\omega+\frac{\Lambda}{3!}e^3+e\mathrm{Tr}(BF_A)+\frac{1}{2\cdot3!}e^3\mathrm{Tr}(B,B)\Big)\\[3pt]
		R_\tau & = \int_\Sigma \tau d_\omega e.
	\end{align}
We refer to these as the constraints of the Yang--Mills-Palatini--Cartan theory.
\end{definition}
\begin{theorem}\label{poiss_brack_YM}
    Let $i^*g$ be degenerate. Then the Poisson brackets of the constraints of \cref{constraints_YMPC} read
        \begin{align*}
    &\begin{aligned}
        	&\{L_c, L_c\}  = - \frac{1}{2} L_{[c,c]} & \qquad\qquad
       	& \{M_\mu, M_\mu\}  = - \frac{1}{2} M_{[\mu,\mu]} \\
       	& \{L_c, P^{A}_{\xi}\}  =  L_{\mathcal{L}_{\xi}^{\omega_0}c} &
       	& \{H^{A}_{\lambda},H^{A}_{\lambda}\}  \approx 0 \\
            & \{L_c, M_\mu\}  =  0 &
       	& \{P_\xi,M_\mu\} = M_{\mathcal{L}_\xi^{A_0}\mu} \\        
            & \{H^{A}_\lambda, M_\mu\}  =  0 &
       	& \{R_\tau, M_\mu\}  =  0  \\        
       	&  \{L_c, R_{\tau}\} =  -R_{p_{\mathcal{S}}[c, \tau]} & 
       	& \{P^{A}_{\xi},R_{\tau}\}  = R_{p_{\mathcal{S}}\mathcal{L}_{\xi}^{\omega_0}\tau}\\
       	&\{ R_{\tau}, H^{A}_{\lambda} \}   \approx G_{\lambda \tau}+K_{\lambda\tau}^{A}  &
 		& \{R_{\tau},R_{\tau}\}  \approx F_{\tau \tau}
        \end{aligned}\\
        &\{P^{A}_{\xi}, P^{A}_{\xi}\}  =  \frac{1}{2}P^{A}_{[\xi, \xi]}- \frac{1}{2}L_{\iota_{\xi}\iota_{\xi}F_{\omega_0}}-\frac12M_{\iota_\xi\iota_\xi F_{\omega_0}}\\
	   &\{L_c,  H^{A}_{\lambda}\}
        = - P^{A}_{X^{(a)}} + L_{X^{(a)}(\omega - \omega_0)_a} - H^{A}_{X^{(n)}} \\        
        &\{P^{A}_{\xi},H^{A}_{\lambda}\} =  P^{A}_{Y^{(a)}} -L_{ Y^{(a)} (\omega - \omega_0)_a} +H^{A}_{ Y^{(n)}},
    \end{align*}
	with $X=[c,\lambda e_n]$ and $Y=\mathcal{L}_\xi^{\omega_0}(\lambda e_n)$ and where the superscripts $(a)$ and $(n)$ describe their components with respect to $e_a,e_n$. Furthermore, $F_{\tau\tau}$, $G_{\lambda\tau}$ and $K_{\lambda\tau}^{A}$ are functional of $e,\omega,A,B,\tau$ and $\lambda$ defined in the proof\footnote{$F$ and $G$ are properly defined in \cite{CCT21} (proof of Theorem $30$).} which are not proportional to any other constraint.
\end{theorem}
\begin{proof}
Similarly to the proof of \cref{poiss_brack_scalar}, we will introduce now a split in some of the constraints. In this case, we have
    \begin{align}\label{split_int_YM}
        &P^{A}_\xi = P_\xi+p^{A}_\xi& \,\, &H^{A}_\lambda = H_\lambda+h^{A}_\lambda.
    \end{align}
Moreover, from \cite{CCT21} and \cite{CCR2022}, we have knowledge of the following brackets
    \begin{align*}
        &\begin{aligned}
        	&\{L_c,L_c\} = -\frac{1}{2}L_{[c,c]} & \qquad\qquad
       	& \{L_c,P^{A}_\xi\}   = L_{\mathcal{L}_\xi^{\omega_0}c} \\[2pt]
            &\{P^{A}_\xi,P^{A}_\xi\} = \frac{1}{2}P^{A}_{[\xi,\xi]}-\frac{1}{2}L_{\iota_\xi\iota_\xi F_{\omega_0}}& \qquad\qquad
            & \{L_c,R_\tau\} = -R_{p_\mathcal{S}[c,\tau]}\\[2pt]
            &\{P_\xi,R_\tau\}=R_{p_\mathcal{S}\mathcal{L}_\xi^{\omega_0}\tau}& \qquad\qquad
            &\{ R_{\tau}, H_{\lambda} \}   \approx  G_{\lambda \tau} \\[2pt]
            &\{R_{\tau},R_{\tau}\}  \approx F_{\tau \tau}& \qquad\qquad
            &\{H^A_{\lambda},H^A_{\lambda}\}  \approx 0\\[2pt]
            &\{M_\mu, M_\mu\}  = - \frac{1}{2} M_{[\mu,\mu]}& \qquad\qquad
            &\{L_c, M_\mu\}  =  0\\[2pt]
            & \{P^{A}_\xi,M_\mu\} = M_{\mathcal{L}_\xi^{A_0}\mu}& \qquad\qquad
            &\{M_\mu,H^{A}_\lambda\}=0\\[2pt]
        \end{aligned}\\
	   &\{L_c,H^{A}_\lambda\}=-P^{A}_{X^{(a)}}+L_{X^{(a)}(\omega-\omega_0)_a}-         H^{A}_{X^{(n)}}\\[2pt]
        &\{P_\xi,H_\lambda^{A}\}=P^{A}_{Y^{(a)}}-L_{Y^{(a)}(\omega-\omega_0)_a}+H^{A}_{Y^{(n)}},
    \end{align*}
with $F$ and $G$ non-identically-vanishing functional of $\tau$ and $\lambda$ defined in \cite{CCT21} (Theorem $30$), $X=[c,\lambda e_n]$ and $Y=\mathcal{L}_\xi^{\omega_0}(\lambda e_n)$. We are thus left with computing the remaining brackets.

Equivalently to the scalar case, the Hamiltonian vector fields for $R_\tau$ are given by
	\begin{align}
		e\mathbb{R}_e & = [\tau,e] & e\mathbb{R}_\omega & = g(\tau,\omega,e)+d_\omega\tau \\
		\mathbb{R}_A & = 0 & \mathbb{R}_\rho & = 0,
	\end{align}
since in does not possess any variation along the gauge fields. We consider now the variation
    \begin{align}
		\delta M_\mu & = \int_\Sigma \mathrm{Tr}(\mu\delta(d_A\rho)) = \int_\Sigma \mathrm{Tr}(-\mu([\delta A,\rho]+d_A(\delta\rho)) \\
		& = \int_\Sigma \mathrm{Tr}\big([\mu,\rho]\delta A+d_A\mu\,\delta\rho\big),
	\end{align}
	and therefore we obtain the following Hamiltonian vector fields
	\begin{align}
		\mathbb{M}_e & = 0 & \mathbb{M}_\omega & = 0 \\
		\mathbb{M}_\rho & = [\mu,\rho] & \mathbb{M}_A & = d_A\mu.
	\end{align}
From \cite{CCR2022}, for $p^{A}_\xi$, we have
	\begin{align}
		\mathbb{p}_e^{A} & = 0 & \mathbb{p}_\omega^{A} & = 0 \\
		\mathbb{p}_\rho^{A} & = -\mathcal{L}_\xi^{A_0}\rho & \mathbb{p}_A^{A} & = -\mathcal{L}_\xi^{A_0}(A-A_0)-\iota_\xi F_{A_0},
	\end{align}
whereas, for $h^{A}_\lambda$, the Hamiltonian vector fields read
        \begin{align}
            & \mathbb{h}_e^{A} = 0&
            &\mathbb{h}_\rho^{A} = d_A(\lambda e_neB)\\
            &\mathbb{h}_A^{A} = \lambda(B,e_ne)&
            &e\mathbb{h}_\omega^{A} = \mathrm{Tr}\big(\lambda e_nBF_A+\lambda e_n\frac{e^2}{4} (B,B)-\lambda eB(B,e_ne)\big).
        \end{align}
Now, we are left with computing the Poisson brackets of the constraints for $\{R_\tau,h_\lambda^{A}\}$, $\{R_\tau,p_\xi^{A}\}$ and $\{R_\tau,M_\mu\}$. We start with noticing that
\begin{align}
    \{R_\tau,p^{A}_\xi\}=\{R_\tau,M_\mu\}=0
\end{align}
since both $p_\xi^{A}$ and $M_\mu$ have vanishing Hamiltonian vector fields along $e$ and $\omega$. Then, we are left with computing
    \begin{align}
		\{R_\tau,h_\lambda^{A}\} &= \int_\Sigma \mathrm{Tr}\Big(\lambda e_nBF_AW_1^{-1}[\tau,e]+\lambda e_n\frac{e}{4}(B,B)[\tau,e]\\
        &\phantom{= \int_\Sigma}-\lambda B(B,e_ne)[\tau,e]\Big),
	\end{align}
where the second term is zero because of $e[\tau,e]=0$ and the first and third terms in general do not vanish. In fact, we have
	\begin{align}
		&\{R_\tau,h_\lambda^{A}\} =\\[4pt]
        &=\int_\Sigma \mathrm{Tr}\Big(\lambda e_nBF_AW_1^{-1}[\tau,e]-\lambda B(B,e_ne)[\tau,e]\Big) \\[4pt]
		&  = \int_\Sigma \mathrm{Tr}\Big(\frac{\lambda e_n}{2}B(B,e^2)-\lambda B(B,e_ne)e\Big)W_1^{-1}[\tau,e] \\[4pt]
		&  = \int_\Sigma \mathrm{Tr}\bigg(\lambda B\Big(\frac{e_n}{2}(B,e^2)-(B,e_ne)e\Big)\bigg)W_1^{-1}[\tau,e]	 \\[4pt]
		& \approx\colon K_{\lambda\tau}^{A},
	\end{align}
where $W_1^{-1}\colon\Omega_\Sigma^{2,2}\to\Omega_\Sigma^{1,1}$ indicates the inverse of the map $W_1^{\Sigma,(1,1)}$ and the symbol $\approx\colon$ means that we are defining the quantity $K_{\lambda\tau}^{A}$ on the constraint submanifold. Then, thanks to Corollary $12$ of \cite{CCT21}, we can write the explicit form of $K_{\lambda\tau}^{A}$ by means of the independent components $\mathcal X$ and $\mathcal Y$ of $\tau$, defined in Proposition $8$ of\cite{CCT21}. Hence, we define $ K_{\lambda\tau}^{A}$ as
\begin{align}\label{dec_KA}
    K_{\lambda\tau}^{A}&=\int_\Sigma\mathrm{Tr}\bigg(\lambda\Big( \frac{1}{2}\big(\sum_{\mu=1}^{2} \mathcal{Y}_{\mu} \mathcal C_\mu^\mu-\sum_{\mu_1\neq\mu_2=1}^{2} \mathcal{X}_{\mu_1}^{\mu_2}\mathcal C^{\mu_1}_{\mu_2}\big)\\
    &\phantom{=\int_\Sigma}-\big(\sum_{\mu=1}^{2} \mathcal{Y}_{\mu} \mathcal D_\mu^\mu-\sum_{\mu_1\neq\mu_2=1}^{2} \mathcal{X}_{\mu_1}^{\mu_2}\mathcal D^{\mu_1}_{\mu_2}\big)\Big)\bigg),
\end{align}
where $\mathcal C_\sigma^\rho\coloneqq (B^{\rho3}-B^{\rho4})(B,e^2)_{3\sigma}$ and $\mathcal D_\sigma^\rho\coloneqq (B^{\rho3}-B^{\rho4})(B,e_ne)_{\sigma}$.

Therefore, thanks to the linearity of the Poisson brackets together with the known results, this completes the proof.
\end{proof}
\subsection{Spinor field} The concept of a spinor field is central in mathematical physics. The idea of a spinor field is funded on the definition a particular subalgebra of the tensor algebra over a vector space, called the \textit{Clifford algebra}.
In the following, we will recall the basic and fundamental results about the structure of these algebras in order to be able to write the Palatini--Cartan theory coupled with a Dirac spinor.
\begin{definition}\label{def_cliffalg}
    Let, $V$ be a vector space over $\mathbb K=\mathbb R, \mathbb C$ and $g\colon V\times V\to \mathbb K$ be a symmetric bilinear form.\footnote{We call this space a \textit{quadratic vector space}.} Moreover, let $I_g$ be the two sided ideal in the tensor graded algebra $T(V)$ of $V$ generated by
    \begin{align}
        \{v\otimes v+g(v,v)1, v\in V\},
    \end{align}
    where $1\in T(V)$ is the unit element. Then, we define the Clifford algebra $\mathrm{Cl}(V,g)$ as the filtered algebra given by the quotient
    \begin{align}
        \mathrm{Cl}(V,g)\coloneqq \frac{T(V)}{I_g}.
    \end{align}
\end{definition}
\begin{remark}
    The general definition of a Clifford algebra is given by means of a universal property in the category of unital associative algebras. One can recover \cref{def_cliffalg} by building a functor between the category of vector spaces endowed with a symmetric bilinear form and the category of unital associative algebras. Then, the universal property guarantees that morphisms extend uniquely to Clifford algebras homomorphisms.
\end{remark}
In the following, we will state some results which are well-known facts in the literature. They will serve as a basis in order to build the theory of spin coframes, which can be regarded as a sort of generalization of the vielbein and the coframe formalism. We refer to \cite{wernli2019}, \cite{fatibene2018} and references therein for the proofs of these results as well as more details.
\begin{definition}
    Let $V$ be a quadratic vector space on $\mathbb R$ and let $(p,q)$ be the signature of $g$. Moreover, let $\mathrm{Cl}^{+}(V,g)\coloneqq \mathrm{Cl}^0\oplus \mathrm{Cl}^2\oplus \mathrm{Cl}^4\oplus...$ be the subalgebra defined by the even grading. We define the group $\mathrm{Pin}_{p,q}\subset \mathrm{Cl}(V,g)$ as the subgroup of the group of units in $\mathrm{Cl}(V,g)$ generated by $v\in V$ such that $|g(v,v)|=1$.\\
    Then, we defined the group $\mathrm{Spin}_{p,q}$ as the subgroup of $\mathrm{Pin}_{p,q}$ given by
    \begin{align}
        \mathrm{Spin}_{p,q}\coloneqq \mathrm{Pin}_{p,q}\cap \mathrm{Cl}^{+}(V,g).
    \end{align}
\end{definition}
\begin{proposition}\label{spin_un_cov}
    Let $V$ be a quadratic vector space on $\mathbb R$ and let $(p,q)$ be the signature of $g$. Moreover, let $\rho\colon\mathrm{Spin}_{p,q}\to \mathrm{GL}\big(\mathfrak{spin}_{p,q}\big)$ be the adjoint representation. Then, we have the following:
    \begin{enumerate}
        \item[--] $\mathfrak{spin}_{p,q}\subset\mathrm{Cl}(V,g)$;\\\vspace{-0.25cm}
        \item[--] The map $\rho$ acts as $\mathrm{SO}(p,q)$ on $V$\footnote{Here, we regard $V$ as a first grade subspace of the Clifford algebra.} (or, for its complexification, as $\mathrm{SO}(n)\times \mathrm{U}(1)$ with $n=p+q$);\\\vspace{-0.25cm}
        \item[--] The map $\rho$ defines a covering map\footnote{By abuse of notation, we denote the covering map and the adjoint representation in the same manner.} $\rho\colon\mathrm{Spin}_{p,q}\to\mathrm{SO}(p,q)$.
    \end{enumerate}
    Furthermore, the group $\mathrm{Spin}_{p,q}$ is simply connected and it is the universal cover of $\mathrm{SO}(p,q)$. Therefore, in particular, $\mathrm{Spin}_{3,1}\cong\mathrm{SL}(2,\mathbb C)$.
\end{proposition}
\begin{definition}\label{spinbeindef}
	 Let $\hat P$ be a principal $\mathrm{Spin}_{p,q}$-bundle on $M$ and $LM$ the frame bundle. Then, we define the \emph{spin map} $E\colon \hat P\to LM$ as the principal bundle morphism such that the following diagram commutes
\begin{center}
\begin{tikzcd}
		\hat{P} \arrow[rdd, "\hat{p}", bend right] \arrow[rr, "E"] \arrow[rd, "\rho"] &                                  & LM \arrow[ldd, "\pi"', bend left] \\
		& P \arrow[d, "p"] \arrow[ru, "\tilde e"] &                                   \\
		& M                                &                                  
	\end{tikzcd}
	\end{center}
where $\rho\colon \hat P\to P$ denotes the bundle morphism induced by the covering map of \cref{spin_un_cov} and $\tilde e$ the vielbein of \cref{coframedef}.
\end{definition}
    The following result will be a particular example of the broader spectrum of the classification of Clifford algebras. In a nutshell, they exhibit a $2$-periodicity in the complex case and a $8$-periodicity in the real case.
\begin{theorem}\label{cliff_isom}
   Let $V$ be a $4$-dimensional quadratic vector space on $\mathbb K$ and, in particular, if $\mathbb K=\mathbb R$, let $(p,q)=(3,1)$. Furthermore, let $\mathrm{M}_{4\times 4}(\mathbb K)_{Cl}$ denote the algebra of $4\times4$ matrices on $\mathbb K$ endowed with the Clifford structure. Then, we have the following isomorphism
   \begin{align}
       \lambda\colon\mathrm{Cl}(V,g)\to \mathrm{M}_{4\times 4}(\mathbb K)_{Cl}.
   \end{align}
\end{theorem}
\begin{remark}
    If we consider the complexification of the algebra $\mathfrak{spin}_{3,1}^{\mathbb C}$, as a consequence of \cref{cliff_isom}, the adjoint representation $\rho\colon\mathrm{Spin}_{3,1}\to \mathrm{GL}\big(\mathfrak{spin}_{3,1}^{\mathbb C}\big)$ can be regarded as acting on $\mathrm{M}_{4\times 4}(\mathbb C)_{Cl}$, since $\mathfrak{spin}_{p,q}\subset\mathrm{Cl}(V,g)$. Moreover, we know by \cref{spin_un_cov} that $\rho$ acts as $\mathrm{SO}(3,1)$ on $V$. Hence, this statement takes the form
    \begin{align}
        \rho_S (\gamma^a)=S\gamma^a S^{-1} =\Lambda^a_b \gamma^b,
    \end{align}
    where  $S\in\lambda\big(\mathrm{Spin}_{3,1}\big)$ and $\Lambda\in\mathrm{SO}(3,1)$ is the matrix associated to $S$ under the covering map with $a,b=1,2,3,4$. In other words, the complexified algebra of the spin group, where the adjoint representation acts, can be expressed in terms of $\gamma$-matrices, which can be also labeled according to a basis of $V$, i.e. $\gamma=\gamma^av_a\in V\otimes \mathrm{M}_{4\times 4}(\mathbb C)_{Cl}$, such that the Clifford relation reads
    \begin{align}
        \{\gamma^a,\gamma^b\}=-2\eta^{ab}1_{4\times4},
    \end{align}
    where the brackets denote the anti-commutators.\\
    Furthermore, if we denote with $\displaystyle{f\colon \mathrm{SO}(3,1)\to \mathrm{Aut}(V)}$ the fundamental representation of $\mathrm{SO}(3,1)$, by composition with the adjoint representation of $\mathrm{Spin}_{3,1}$, we can construct the Minkowski bundle as the associated vector bundle to $\hat P$ under the composition, i.e. 
    \begin{align}
    \mathcal V\coloneqq \hat P\times_{f\circ \rho} V.    
    \end{align}
    \end{remark}
    Note that the isomorphism of \cref{cliff_isom} defines a representation of the complexified group $\mathrm{Spin}_{3,1}^{\mathbb C}$ on $\mathbb C^4$. This representation is called the $\gamma$-representation and it corresponds to the representation $(\frac12,0)\oplus(0,\frac12)$ of $\mathrm{SL}(2,\mathbb C)$ (thanks to the group isomorphism $\mathrm{Spin}_{3,1}\cong\mathrm{SL}(2,\mathbb C)$). This fact allows to have the following definition.
    \begin{definition}
        Let $\gamma\colon\mathrm{Spin}_{3,1}^{\mathbb C}\to \mathrm{Aut}(\mathbb C^{4})$ be the  $\gamma$-representation of the spin group. Then, we define the \emph{spinor bundle} as the associated vector bundle to $\hat P$ under $\gamma$, namely
        \begin{align}
            S\coloneqq \hat P\times_{\gamma}\mathbb C^4.
        \end{align}
        We define a \emph{spinor field}\footnote{In our case, we will only deal with Dirac spinors. Therefore, the term “spinor" refers uniquely to a Dirac one. In a more general setting, we must slightly generalize our definition in order to include other spin structures.} as a section of the odd-bundle $\Pi S$, where $\Pi$ indicates the parity reversal operation\footnote{Parity inversion is fundamental since we want spinors to be Grassmannian/odd quantities.}.
    \end{definition}
    \begin{remark}\label{rem_gamma_matr1}
        Notice that, in this context, we can regard the $\gamma$-matrices as elements $\gamma\in\Gamma\big(\mathcal V\otimes \mathrm{End}(\Pi S)\big)$. Note also that, by construction, the parity of a spinor field $\psi\in\Gamma(\Pi S)$ is given by $|\psi|=1$.
    \end{remark}
\begin{proposition}\label{iso_spin_wedge2}
    Given a real vector space $V$ and the isomorphism $\mathfrak{so}(3,1)\cong\bigwedge^2 V$, we have the following algebra isomorphism
    \begin{align}
        \mathrm{d}\rho\colon \mathfrak{spin}_{3,1}\to\bigwedge^2 V,
    \end{align}
    which is given by
    \begin{align}
        \mathrm{d}\rho^{-1}(v\wedge w)=-\frac14[\tilde v,\tilde w],
    \end{align}
    where $v,w\in V$, $\tilde v,\tilde w\in\mathfrak{spin}_{3,1}$ and $\rho\colon\mathrm{Spin}_{3,1}\to \mathrm{GL}\big(\mathfrak{spin}_{3,1}\big)$ is the adjoint representation.
\end{proposition}
If we consider the complexified Lie algebra $\mathfrak{spin}_{3,1}^{\mathbb C}$ and the isomorphism of \cref{iso_spin_wedge2}, we can build a covariant derivative for spinor fields in terms of local connections in $\Omega^{1,2}$. Explicitly, it reads
\begin{align}
    d_\omega\psi=d\psi+[\omega, \psi]=d\psi-\frac14\omega^{ab}\gamma_a\gamma_b\psi.
\end{align}
We define the covariant derivative for the conjugate of $\psi$ such that $d_\omega \overline{\psi}=\overline{d_\omega \psi}$. Therefore, we have
\begin{align}
	d_\omega \overline{\psi}=d\overline{\psi}+ [\omega,\overline{\psi}]= d\overline{\psi} - \frac{1}{4}\omega^{ab}\overline{\psi}\gamma_a \gamma_b.
\end{align}

By \cref{rem_gamma_matr1}, we can extend the definition of the covariant derivative also to the $\gamma$-matrices. It follows the upcoming lemma.
\begin{lemma}
    Let $\gamma\in\Gamma\big(\mathcal V\otimes \mathrm{End}(\Pi S)\big)$. Then, it holds
    \begin{align}
        d_\omega \gamma=0.
    \end{align}
\end{lemma}
\begin{proof}
    See \cite{CCR2022}.
\end{proof}
The space of fields of the Spinor--Palatini--Cartan theory is given by\footnote{Where $\bar S$ is simply given by the conjugate representation $\bar\gamma$.}
    \begin{align}
    \mathcal F^{\psi}=\tilde\Omega^{1,1}\times \mathcal A(P)\times \Gamma(\Pi S)\times\Gamma(\Pi \bar S)\ni (e,\omega,\psi,\overline\psi),
    \end{align}
whereas the action functional reads
    \begin{align}\label{eq:actionPCT_spinor}
    \mathcal S^{\psi}=\mathcal S_{PC}+\int_{M}\frac{i}{12}e^3( \overline{\psi} \gamma d_\omega \psi - d_\omega \overline{ \psi}\gamma \psi ).
    \end{align}
    It follows that the field equations are the following Euler-Lagrange equations for the action $\mathcal S^{\psi}$
    \begin{align}
	eF_\omega + \frac{i}{4}e^2( \overline{\psi}\gamma d_\omega\psi - d_\omega \overline{\psi}\gamma \psi)&=0 \\
	ed_\omega e +  \frac{i}{6}(\overline{\psi}\gamma[e^3,\psi] - [e^3,\overline\psi]\gamma\psi)&=0 \\
	\frac{e^{3}}{6}\gamma d_\omega \psi - \frac{1}{12}d_\omega e^3\gamma \psi   &=0 \\
	\frac{e^{3}}{6}d_\omega \overline{\psi}\gamma +  \frac{1}{12}d_\omega e^3\overline{\psi}\gamma  &=0,
\end{align}
where we define, for $X\in \Gamma(\mathcal V)$ and $\alpha\in \Omega^{r,k}_\Sigma$, the contraction
\begin{align}
    \iota_X \alpha\coloneqq\frac{\eta_{ab}}{(k-1)!} X^a \alpha^{b {i_2}\cdots {i_k}} 	v_{i_2}\wedge \cdots \wedge v_{i_k}
\end{align}
and consequently, for $\chi \in \Omega^{i,j}_\Sigma$, the brackets
\begin{align}
\begin{cases}
    [\chi,\psi]&\coloneqq\frac{1}{4(j-1)}\iota_{\gamma} \iota_{\gamma} \chi \psi\\[5pt]
    [\chi,\overline\psi]&\coloneqq-\frac{(-1)^{|\chi||\psi|}}{4(j-1)}\overline\psi\iota_{\gamma} \iota_{\gamma} \chi,
\end{cases}
\end{align}
where $|\chi|$ is the parity of $\chi$ and $|\psi|$ the parity of $\psi$.

Similar to the preceding sections, the space of pre-boundary fields $\tilde{\mathcal F}^{\psi}_\Sigma$, as defined in \cref{PC_preboundary} for the Palatini--Cartan theory, can be established by pulling back the fields to the boundary $\Sigma$. Furthermore, we will keep denoting the fields on the boundary in the same way as those in the bulk.

As outlined in \cite{CCR2022}, we can now define the geometric phase space of the theory through a reduction using the kernel of the pre-symplectic form.
\begin{theorem}\label{mastertheorem_spinor}
The geometric phase space for the Spinor--Palatini--Cartan theory is the symplectic manifold $(\mathcal F^{\psi}_\Sigma,\varpi)$ given by the following equivalence relations on the space of pre-boundary fields $\tilde{\mathcal F}^{\psi}_\Sigma$
\begin{align}
\omega^\prime\sim\omega\quad&\Longleftrightarrow\quad\omega^\prime-\omega \in\mathrm{Ker}W_{1}^{\Sigma, (1,2)}
\end{align}
and the symplectic form
    \begin{align}
       \varpi = \int_{\Sigma} e\delta e \delta \omega + i\frac{e^2}{4}\left( \overline{\psi}\gamma \delta \psi - \delta \overline{\psi}	\gamma \psi \right)\delta e + i \frac{e^3}{ 3!}  \delta \overline{\psi} \gamma \delta \psi.
    \end{align}
We denote this equivalence class as $\mathcal A(i^* P)_{red}$.
\end{theorem}
\begin{proof}
    See \cite{CCR2022}.
\end{proof}
\begin{remark}
    Likewise the preceding cases, we notice that the equivalence class of $\omega$, defining the geometric phase space, remains equal to the Palatini--Cartan theory. In fact, similarly to the previous couplings, this can be seen as a consequence of the fact that the symplectic form does not have any other piece along $\omega$, but the one equal to the Palatini--Cartan case.
\end{remark}
\begin{remark}
    In the case at hand, the field equations see a substantial difference. Namely, the Levi-Civita (or torsion-free) condition $e d_\omega e=0$ no longer holds. Indeed, the Lagrangian of the theory couples the connection with the spinor. Therefore, the structural and the degeneracy constraints take the form
\begin{equation}\label{str_constr_spin}
        \begin{cases}  e_n (\alpha_\psi- p_{\mathcal{T}}\alpha_\psi) \in \Ima W_{1}^{\Sigma,(1,1)}\\[6pt]
        p_{\mathcal T}\alpha_\psi=0.
        \end{cases}
    \end{equation}
with
\begin{align}
    \alpha_\psi\coloneqq d_\omega e+\frac{i}{4}(\overline\psi\gamma[e^2,\psi]-[e^2,\overline\psi]\gamma\psi).
\end{align}
\end{remark}
The following proposition will ensure that, although the form of $\alpha_\psi$ is sensibly different from the preceding cases, the form of the functional $R^\psi_\tau$ will coincide with the one of the Palatini--Cartan theory.
\begin{proposition}\label{trivial_Rpsi}
    Let $\tau\in\mathcal S$. Then, we have the following identity
    \begin{align}
        \tau(\overline{\psi} \gamma [e^2,\psi]-[e^2,\overline{\psi}]\gamma \psi)=0.
    \end{align}
\end{proposition}
\begin{proof}
    The proof comes by applying twice \cref{lemma_switch_bracket}. Therefore, by means of \cref{cortaue_n}, we have
    \begin{align}
        \hspace{-1cm}\tau(\overline{\psi} \gamma [e^2,\psi]-[e^2,\overline{\psi}]\gamma \psi)&=e_n\beta(\overline{\psi} \gamma [e^2,\psi]-[e^2,\overline{\psi}]\gamma \psi)\\[3pt]
        &=e_ne^2(\overline{\psi} \gamma [\beta,\psi]-[\beta,\overline{\psi}]\gamma \psi)\\[3pt]
        &=e\beta(\overline{\psi} \gamma [e_ne,\psi]-[e_ne,\overline{\psi}]\gamma \psi)\\[3pt]
        &=0,
    \end{align}
    since $\beta\in\mathrm{Ker}W_1^{\Sigma, (1,2)}$.
\end{proof}
We are now able to properly give the constraints of the theory.
\begin{definition}\label{constraints_spinor}
Let $c\in\Omega_\Sigma^{0,2}[1]$, $\xi\in\mathfrak{X}(\Sigma)[1]$, $\lambda\in C^\infty(\Sigma)[1]$ and $\tau\in\mathcal{S}[1]$. Then, we define the following functionals
    \begin{align}
		L^\psi_c & = \int_\Sigma ced_\omega e- i \frac{e^3}{2\cdot 3!} \left( [c,\overline{\psi}]\gamma \psi - \overline{\psi} \gamma [c,\psi] \right) \\[3pt]
		P^\psi_\xi & = \int_\Sigma \frac{1}{2}\iota_\xi(e^2)F_\omega+\iota_\xi(\omega-\omega_0)ed_\omega e- i \frac{e^3}{2\cdot 3!} \left( \overline{\psi} \gamma \mathrm{L}_{\xi}^{\omega_0}(\psi) - \mathrm{L}_\xi^{\omega_0}(\overline{\psi})\gamma \psi  \right) \\[3pt]
		H^\psi_\lambda & = \int_\Sigma \lambda e_n\Big(eF_\omega+\frac{\Lambda}{3!}e^3+i\frac{e^2}{4}\left( \overline{\psi}\gamma d_\omega \psi - d_\omega\overline{\psi}\gamma \psi\right)\Big)\\[3pt]
		R^\psi_\tau &= \int_\Sigma \tau d_\omega e.
	\end{align}
We refer to these as the constraints of the Spinor--Palatini--Cartan (degenerate) theory.
\end{definition}
\begin{theorem}\label{poiss_brack_spinor}
    Let $i^*g$ be degenerate. Then, the Poisson brackets of the constraints of \cref{constraints_spinor} read
        \begin{align*}
    &\begin{aligned}
        	&\{L^\psi_c, L^\psi_c\}  = - \frac{1}{2} L_{[c,c]} & \qquad\qquad
       	& \{P^\psi_{\xi}, P^\psi_{\xi}\}  =  \frac{1}{2}P^\psi_{[\xi, \xi]}- \frac{1}{2}L_{\iota_{\xi}\iota_{\xi}F_{\omega_0}} \\
       	& \{L^\psi_c, P^\psi_{\xi}\}  =  L^\psi_{\mathcal{L}_{\xi}^{\omega_0}c} &
       	& \{H^\psi_{\lambda},H^\psi_{\lambda}\}  \approx 0 \\
       	&  \{L^\psi_c, R^\psi_{\tau}\} =  -R^\psi_{p_{\mathcal{S}}[c, \tau]} & 
       	& \{R_{\tau}^\psi,P^\psi_{\xi}\}  = R^\psi_{p_{\mathcal{S}}\mathcal{L}_{\xi}^{\omega_0}\tau}\\
       	&\{ R^\psi_{\tau}, H^\psi_{\lambda} \}   \approx G_{\lambda\tau}+K_{\lambda\tau}^\psi &
 		& \{R^\psi_{\tau},R^\psi_{\tau}\}  \approx F_{\tau \tau}
        \end{aligned}\\
	&\{L^\psi_c,  H^\psi_{\lambda}\}
        = - P^\psi_{X^{(a)}} + L^\psi_{X^{(a)}(\omega - \omega_0)_a} - H^\psi_{X^{(n)}}\\        
        &\{P^\psi_{\xi},H^\psi_{\lambda}\} =  P^\psi_{Y^{(a)}} -L^\psi_{ Y^{(a)} (\omega - \omega_0)_a} +H^\psi_{ Y^{(n)}},
    \end{align*}
	with $X=[c,\lambda e_n]$, $Y=\mathcal{L}_\xi^{\omega_0}(\lambda e_n)$ and where the superscripts $(a)$ and $(n)$ describe their components with respect to $e_a,e_n$. Furthermore, $F_{\tau\tau}$, $G_{\lambda\tau}$ and $K_{\lambda\tau}^\psi$ are functionals of $e,\omega,\psi, \overline\psi,\tau$ and $\lambda$ defined in the proof which are not proportional to any other constraint.
\end{theorem}
\begin{proof}
First, we notice that the contraction of the symplectic form with a vector field $\mathbb X\in\mathfrak{X}(\mathcal F_\Sigma^\psi)$ is given by
    \begin{align}\label{contract_symp_spin}
        \iota_{\mathbb{X}}\varpi&=\int_\Sigma e \mathbb{X}_e \delta \omega +\left[ e\mathbb{X}_\omega +\ \frac{i}{4}e^2(  \overline{\psi} \gamma \mathbb{X}_\psi -  \mathbb{X}_{\overline{\psi}}\gamma \psi )\right]\delta e \\
        &\phantom{=\int_\Sigma} + i\delta \overline{\psi} \left( -\frac{e^2}{4}\gamma \psi \mathbb{X}_e + \frac{e^3}{3!}\gamma \mathbb{X}_\psi\right) + i \left( \frac{e^2}{4}\overline{\psi}\gamma \mathbb{X}_e + \frac{e^3}{3!}\mathbb{X}_{\overline{\psi}}\gamma \right)\delta \psi.
    \end{align}
    Then, we start giving the Hamiltonian vector fields of the constraints. For $L^\psi_c$ and $P^\psi_\xi$, from \cite{CCR2022}, we have
    \begin{align}
        &\mathbb{L}^\psi_e = [c,e] &\mathbb{L}^\psi_\psi=[c,\psi] \\  &\mathbb{L}^\psi_\omega = d_{\omega} c &\mathbb{L}^\psi_{\overline{\psi}}=[c,\overline{\psi}]\\
			&\mathbb{P}^\psi_e = - \mathcal{L}_{\xi}^{\omega_0} e &\mathbb{P}^\psi_\psi=-\mathcal{L}_\xi^{\omega_0}(\psi)   \\  &\mathbb{P}^\psi_\omega = - \mathcal{L}_{\xi}^{\omega_0} (\omega-\omega_0) - \iota_ {\xi}F_{\omega_0}   &\mathbb{P}^\psi_{\overline{\psi}}=-\mathcal{L}_{\xi}^{\omega_0}(\overline{\psi}).
    \end{align}
    Whereas, for $H^{\psi}_\lambda$, we have
    \begin{align}
        &\mathbb{H}^{\psi}_e = d_\omega (\lambda e_n ) + \lambda \sigma + \frac{i}{4}\lambda\overline{\psi}\left(\iota_{\gamma}\iota_{\gamma} e_ne \gamma - \gamma \iota_{\gamma}\iota_{\gamma} e_ne   \right)\psi\\[4pt]
	        &e \mathbb{H}^{\psi}_\omega = \lambda e_n \left( F_\omega + \frac{\Lambda}{2}e^2 \right) - i \frac{\lambda e_n}{4} e (\overline{\psi}\gamma d_\omega \psi - d_\omega\overline{\psi} \gamma \psi)\\[4pt]
	         &\frac{e^3}{3!}\gamma \mathbb{H}^{\psi}_{\psi}= \frac{\lambda e_n}{2} e^2 \gamma d_\omega \psi - \frac{\lambda e_n}{4} ed_\omega e \gamma \psi\\
            &\phantom{=\frac{e^3}{3!}\gamma \mathbb{H}^{\psi}_{\psi}}+ \frac{i}{64} \lambda e \left[\overline{\psi}\left( \iota_{\gamma}\iota_{\gamma}(e_n e^2)\gamma - \gamma \iota_{\gamma} \iota_{\gamma}(e_n e^2) \right)  \psi\right]\gamma\psi\\[4pt]
	          &\frac{e^3}{3!}\mathbb{H}^{\psi}_{\overline{\psi}} \gamma = \frac{\lambda e_n}{2} e^2  d_\omega \overline{\psi} \gamma + \frac{\lambda e_n}{4} ed_\omega e \overline{\psi}  \gamma\\
            &\phantom{=\frac{e^3}{3!}\mathbb{H}^{\psi}_{\overline{\psi}} \gamma }- \frac{i}{64} \lambda e \overline{\psi}\gamma\left[\overline{\psi}\left( \iota_{\gamma}\iota_{\gamma}(e_n e^2)\gamma - \gamma \iota_{\gamma} \iota_{\gamma}(e_n e^2) \right)  \psi\right],
    \end{align}
    where $\sigma\in\Omega^{1,1}_\Sigma$. Lastly, the Hamiltonian vector fields of $R_\tau^\psi$, are given by
     \begin{align}
            e\mathbb{R}^\psi _e  &=[\tau,e] \\
             e\mathbb{R}^\psi_\omega  &= \frac{\delta\tau}{\delta e} d_\omega e+d_\omega\tau\\
            \mathbb{R}^\psi_{\psi}&=\mathbb R^\psi_{\overline\psi}=0,
     \end{align}
     since they coincide with the ones of the Palatini--Cartan theory of \cref{constraints_gravity-deg}. Notice that, instead of using the function $g=g(\tau, e,\omega)$, we preferred expressing the variation of $\tau$ with respect to $e$ by means of the functional derivative $\displaystyle{\frac{\delta\tau}{\delta e}}$. However, we have the relation
\begin{align}
    g(\tau, e,\omega)=\frac{\delta\tau}{\delta e}d_\omega e.
\end{align}
    Now, we are ready to compute the Poisson brackets of the constraints. From \cite{CCR2022}, we have already knowledge of the following Poisson brackets
    \begin{align}
        & \{P^\psi_{\xi}, P^\psi_{\xi}\}  =  \frac{1}{2}P^\psi_{[\xi, \xi]}- \frac{1}{2}L^\psi_{\iota_{\xi}\iota_{\xi}F_{\omega_0}} &
        &\{H^{\psi}_\lambda,H^{\psi}_\lambda\}=0\\
		&\{L^\psi_c, P^\psi_{\xi}\}  =  L^\psi_{\mathcal{L}_{\xi}^{\omega_0}c} &
        &\{L^\psi_c, L^\psi_c\} = - \frac{1}{2}L^\psi_{[c,c]}\\
        & \{L^\psi_c,  H^{\psi}_{\lambda}\}  = - P^\psi_{X^{(a)}} + L^\psi_{X^{(a)}(\omega - \omega_0)_a} - H^{\psi}_{X^{(n)}}&
        &\\
		& \{P^\psi_{\xi},H^{\psi}_{\lambda}\}  =  P^\psi_{Y^{(a)}} -L^\psi_{ Y^{(a)} (\omega - \omega_0)_a} +  H^{\psi}_{ Y^{(n)}},
    \end{align}
    with $X= [c, \lambda e_n ]$ and $Y = \mathcal{L}_{\xi}^{\omega_0} (\lambda e_n)$ as above. Therefore, we are left with computing the remaining constraints. First, we notice that
    \begin{align}
        \{R_\tau^\psi,L^\psi_c\}=\{R_\tau,L_c\}=-R_{p_{\mathcal S}[c,\tau]}=-R^\psi_{p_{\mathcal S}[c,\tau]}.
    \end{align}
Similarly, we can also compute the bracket
\begin{align}
    \{R_\tau^\psi,P^\psi_\xi\}=\{R_\tau,P_\xi\}=R_{p_{\mathcal S}\mathcal L_\xi^{\omega_0}\tau}=R^\psi_{p_{\mathcal S}\mathcal L_\xi^{\omega_0}\tau}.
\end{align}
Now, we move on to compute the brackets $\{R^\psi_\tau,R^\psi_\tau\}$ and $\{R^\psi_\tau,H^\psi_\lambda\}$. The first bracket is simply given by
\begin{align}
    \{R^\psi_\tau,R^\psi_\tau\}=\{R_\tau,R_\tau\}\approx F_{\tau\tau}
\end{align}
with $F_{\tau\tau}$ defined in Theorem $30$ of \cite{CCT21} and which is in general non-vanishing on the constraint submanifold. Whereas, for the second one, we obtain

\begin{align}
   \hspace{-1cm}\{R^\psi_\tau,H^{\psi}_\lambda\}&=\int_\Sigma\Big(e_n \frac{\delta\beta}{\delta e}d_\omega e+d_\omega(e_n\beta)\Big)\Big(d_\omega(\lambda e_n)+\lambda\sigma\\
   &\phantom{=\int_\Sigma}-i\lambda(\overline\psi\gamma[e_ne,\psi]-[e_ne,\overline\psi]\gamma\psi)\Big)\\
    &\phantom{=\int_\Sigma}+W_1^{-1}[e_n\beta,e]\big(\lambda e_n(F_\omega+\frac{\Lambda}{2}e^2)\\
    &\phantom{=\int_\Sigma}-\frac{i}{4}\lambda e_n e(\overline\psi\gamma d_\omega\psi-d_\omega\overline\psi\gamma\psi)\big)\\
    &\approx\int_\Sigma-i\lambda\beta d_\omega e_n(\overline\psi\gamma[e_ne,\psi]-[e_ne,\overline\psi]\gamma\psi)\\
    &\phantom{\approx\int_\Sigma}-\frac{i}{4}[e_n\beta,e]\lambda e_n(\overline\psi\gamma d_\omega\psi-d_\omega\overline\psi\gamma\psi)\\
    &\phantom{\approx}+G_{\lambda\tau},
\end{align}
where, in the last passage, we used \cref{lemma_switch_bracket} and the fact that $e_n^2=0$. Moreover, the quantity $G_{\lambda\tau}$ and the map $W_1^{-1}$ are defined respectively in Theorem $30$ of \cite{CCT21} and in the proof of \cref{poiss_brack_YM}. Now, we can notice that, thanks to \cref{lemma_switch_bracket}, we can write
\begin{align}
    &\lambda\beta d_\omega e_n(\overline\psi\gamma[e_ne,\psi]-[e_ne,\overline\psi]\gamma\psi)=\\[3pt]
    &=\lambda e_ned_\omega e_n(\overline\psi\gamma[\beta,\psi]-[\beta,\overline\psi]\gamma\psi)\\[3pt]
    &= \lambda e\beta(\overline\psi\gamma[e_nd_\omega e_n,\psi]-[e_nd_\omega e_n,\overline\psi]\gamma\psi)\\[3pt]
    &=0,
\end{align}
obtaining
\begin{align}
    \{R^\psi_\tau,H^{\psi}_\lambda\}&\approx G_{\lambda\tau}-\int_\Sigma\frac{i}{4}[e_n\beta,e]\lambda e_n(\overline\psi\gamma d_\omega\psi-d_\omega\overline\psi\gamma\psi).
\end{align}
Finally, we can write the integral as
\begin{align}
    \{R^\psi_\tau,H^{\psi}_\lambda\}&\approx G_{\lambda\tau}-\int_\Sigma\frac{i}{4}\lambda\tau[e_n,\hat e](\overline\psi\gamma d_\omega\psi-d_\omega\overline\psi\gamma\psi),
\end{align}
where we implemented again \cref{cortaue_n} and also the relation\footnote{It simply comes from the definition of $\mathcal S$.}
\begin{align}\label{condition_entaue}
    e_n[\tau,e]=\tau[e_n,\hat e]
\end{align}
with $\hat e$ defined as $\hat e\coloneqq e-\tilde e$ (see \cref{maptilderho}). More specifically, using the definition of the independent components of $\tau$, as we did in the proof of \cref{poiss_brack_YM}, we have
\begin{align}
     \{R^\psi_\tau,H^{\psi}_\lambda\}\approx G_{\lambda\tau}+K^\psi_{\lambda\tau},
\end{align}
with
\begin{align}
    K^\psi_{\lambda\tau}\coloneqq -\int_\Sigma i\lambda\Big(\sum_{\mu=1}^{2} \mathcal{Y}_{\mu} \big(\hat g_nd_\omega J_\psi\big)_{\mu }^{3\mu}+\sum_{\mu_1\neq\mu_2=1}^{2} \mathcal{X}_{\mu_1}^{\mu_2}\big(\hat g_nd_\omega J_\psi\big)_{3\mu_2 }^{\mu_1}\Big),
\end{align}
where $\hat g_n\coloneqq[e_n,\hat e]\in\Omega^{1,0}_\Sigma$ and $d_\omega J_\psi\coloneqq d_\omega(\overline\psi\gamma\psi)\in\Omega^{1,1}_\Sigma$.

This final result completes the proof.
\end{proof}

\section{First and second class constraints}\label{sectio_firstsecond}
In \cite{CCT21}, a study of first- and second-class constrains has been presented. In the following section, we will recall the main results and adapt them to the present analysis.

\begin{definition}\label{def_firstsecond}
   Consider a symplectic manifold $\mathcal{F}$ and a set of smooth maps $\phi_i \in C^{\infty}(\mathcal{F})$ defined on it. Let $C_{ij}=\{\phi_i, \phi_j\}$ represent the matrix of Poisson brackets associated with these maps. The count of second-class maps in the set corresponds to the rank of the matrix $C_{ij}$ evaluated at the zero locus defined by the $\phi_i$s\footnote{We assume the rank to be constant on the zero locus.}. In particular, if $C_{ij} \approx 0$, we categorize all the maps as first-class.
\end{definition}

\begin{proposition}\label{prop:numberofsecondclassconstraints}
    Let $\mathcal{F}$ be a symplectic manifold and let $\psi_i,\phi_j \in C^{\infty}(\mathcal{F})$, where $i=1\dots n$ and $j=1\dots m$. Moreover, denote with $C_{jj'}, B_{ij}$ and $ D_{i i'}$ the matrices representing, respectively, the Poisson brackets  $\{\phi_j,\phi_{j'}\}$, $\{\psi_i,\phi_j\}$ and $\{\psi_i,\psi_{i'}\}$, with $i,i'=1\dots n$ and $j,j'=1\dots m$. Then, if $D$ is invertible and {$C= -B^T D^{-1}B$}, the number of second-class constraints is $n$, i.e. the rank of the matrix $D$.   
\end{proposition}
\begin{proof}
    See \cite{CCT21}.
\end{proof}

\begin{theorem}\label{theorem_firstsecond}
    Let the symbol $\bullet$ be such that $\bullet=\phi, A, \psi$. Then, the constraints $L^\bullet_c$, $P^\bullet_{\xi}$, $H^\bullet_{\lambda}$ and $R^\bullet_{\tau}$ do not form a first class system. In particular, $R_{\tau}$ is a second class constraint.
\end{theorem}
\begin{proof}
    The proof follows verbatim the one of \cite{CCT21}.
\end{proof}

We can now determine the degrees of freedom of the reduced phase space. Let $r$ denote the number of degrees of freedom in the reduced phase space, $p$ the number of degrees of freedom in the geometric phase space, $f$ the number of first-class constraints, and $s$ the number of second-class constraints. The relationship among them is given by\footnote{The proof of this formula is contained in \cite{HT}.}
\begin{align}
r = p - 2f - s.
\end{align}
For all the possible couplings, it follows that we obtain the same result of the Palatini--Cartan theory, i.e.,
\begin{align}
    r = 2.
\end{align}
\begin{remark}
    We notice that in the non-degenerate case we would obtain $r=4$. This reflects the existence of the constraint $R^\bullet_\tau$, which has been proven giving rise to a second-class system. We recall that such a constraint was implied by the geometry of the theory. In particular, together with some additional condition, it ensured the possibility of uniquely fixing a representative of the equivalence class of $\omega$. In fact, on a null-boundary, the space $\mathcal T$ defined in \cref{fund_spaces} is non-trivial.\\
    In physics, it is well-known that GR carries four local degrees of freedom.\footnote{Notice that sometimes the literature reports only two degrees of freedom. This is simply a consequence of considering the dimension of the phase space or just the one of the base manifold.} However, the constraint analysis of the degenerate theory sheds light of the fact that these local degrees of freedom, in the case of manifolds with a null-boundary, are reduced to only two. This fact has important implications regarding the study of black-holes, since the event horizon is a null-hypersurface.
\end{remark}

\section*{Appendix}
\begin{lemma}\label{strconstr_free}
Let $e_n\in\Omega_\Sigma^{0,1}$ be as in \cref{strconstr_free_deg} and $\alpha\in\Omega_\Sigma^{2,1}$. Then, we have
$$\alpha = 0$$
if and only if
	\begin{equation}
        \begin{cases}  \alpha\in\mathrm{Ker}W_1^{\Sigma, (2,1)}\\[3pt]
        e_n \alpha \in \Ima W_{1}^{\Sigma,(1,1)}.
        \end{cases}
    \end{equation}
\end{lemma}
\begin{proof}
    See \cite{CCS2020}.
\end{proof}
\begin{corollary}\label{lemma_map_en_0}
Let $e_n\in\Omega_\Sigma^{0,1}$ be as in \cref{strconstr_free_deg} and $\gamma\in\Omega_\Sigma^{2,2}$. Then, we have the unique decomposition
\begin{align}
    \gamma=e\sigma+e_n\alpha,
\end{align}
with $\sigma\in\Omega_\Sigma^{1,1}$ and $\alpha\in\mathrm{Ker}W_1^{\Sigma, (2,1)}$.
\end{corollary}
\begin{proof}
    We define the map
    \begin{align}
        W_1^{n,\Sigma,(i,j)}\colon\Omega_\Sigma^{i,j}&\to\Omega_\Sigma^{i,j+1}\\
        \kappa&\mapsto e_n\kappa.
    \end{align}
From \cref{strconstr_free}, we know that the map $W_1^{n,\Sigma,(2,1)}|_{\mathrm{Ker}W_1^{\Sigma, (2,1)}}$ is injective\footnote{It is easy to see by setting $e_n\alpha=0$.}, whereas, the proof of the injectivity of $W_1^{\Sigma,(1,1)}$ is given in \cite{Can2021}. Moreover, \cref{strconstr_free} basically states that the intersection $\mathrm{Im}W_1^{\Sigma,(1,1)}\cap \mathrm{Im}W_1^{n,\Sigma,(2,1)}|_{\mathrm{Ker}W_1^{\Sigma, (2,1)}}$ is trivial. We then have
\begin{align}
    \mathrm{dim}(\mathrm{Im}W_1^{\Sigma, (1,1)})=\mathrm{dim}(\Omega_\Sigma^{1,1})=12
\end{align}
and
\begin{align}
    \mathrm{dim}(\mathrm{Im}W_1^{n,\Sigma,(2,1)}|_{\mathrm{Ker}W_1^{\Sigma, (2,1)}})=\mathrm{dim}(\mathrm{Ker}W_1^{\Sigma, (2,1)})=6,
\end{align}
since we know from \cite{Can2021} that $W_1^{\Sigma, (2,1)}$ is surjective. Given that
\begin{align}
    \mathrm{dim}(\Omega_\Sigma^{2,2})=18,
\end{align}
it follows the statement.
\end{proof}
\begin{lemma}\label{e_n_inj}
Let $e_n\in\Omega_\Sigma^{0,1}$ be as in \cref{strconstr_free_deg} and $v\in\Omega_\Sigma^{1,2}$. Then, we have
    \begin{align}
    v=0
\end{align}
if and only if
\begin{align}
    \begin{cases}
        v&\in\mathrm{Ker}W_1^{\Sigma, (1,2)}\\[3pt]
        e_n v&\in \mathrm{Im}W_1^{\Sigma, (0,2)}.
    \end{cases}
\end{align}
\begin{proof}
    This statement is the precise analogous of \cref{strconstr_free} and the proof follows verbatim upon the substitution $W_1^{\Sigma, (1,1)}\to W_1^{\Sigma, (0,2)}$. 
\end{proof}
\end{lemma}
\begin{corollary}\label{lemma_map_en}
Let $e_n\in\Omega_\Sigma^{0,1}$ be as in \cref{strconstr_free_deg} and $\theta\in\Omega_\Sigma^{1,3}$. Then, we have the unique decomposition
\begin{align}
    \theta=ec+e_n\beta,
\end{align}
with $c\in\Omega_\Sigma^{0,2}$ and $\beta\in\mathrm{Ker}W_1^{\Sigma, (1,2)}$.
\end{corollary}
\begin{proof}
   Given the map $ W_1^{n,\Sigma,(1,2)}$ defined in \cref{lemma_map_en_0}, from \cref{e_n_inj}, we know that the map $W_1^{n,\Sigma,(1,2)}|_{\mathrm{Ker}W_1^{\Sigma, (1,2)}}$ is injective, whereas, the proof of the injectivity of $W_1^{\Sigma,(0,2)}$ is given in \cite{Can2021}. Moreover, \cref{e_n_inj} basically states that the intersection $\mathrm{Im}W_1^{\Sigma,(0,2)}\cap \mathrm{Im}W_1^{n,\Sigma,(1,2)}|_{\mathrm{Ker}W_1^{\Sigma, (1,2)}}$ is trivial. We then have
\begin{align}
    \mathrm{dim}(\mathrm{Im}W_1^{\Sigma, (0,2)})=\mathrm{dim}(\Omega_\Sigma^{0,2})=6
\end{align}
and
\begin{align}
    \mathrm{dim}(\mathrm{Im}W_1^{n,\Sigma,(1,2)}|_{\mathrm{Ker}W_1^{\Sigma, (1,2)}})=\mathrm{dim}(\mathrm{Ker}W_1^{\Sigma, (1,2)})=6,
\end{align}
since we know from \cite{Can2021} that $W_1^{\Sigma, (1,2)}$ is surjective. Given that
\begin{align}
    \mathrm{dim}(\Omega_\Sigma^{1,3})=12,
\end{align}
it follows the statement.
\end{proof}
\begin{proposition}\label{cortaue_n}
    Let $\tau\in\mathcal{S}$. Then, $\tau=e_n\beta$ with $\beta\in\Omega_\Sigma^{1,2}[1]$ such that $e_n\beta\in\mathrm{Ker}\tilde{\varrho}^{1,3}$ and $e_n$ defined as above.
\end{proposition}
\begin{proof}
    From \cref{lemma_tau_1}, in particular, we have that
    \begin{align}\label{implicatio_lemma_tau_1}
        p_{\mathcal T}\alpha=0\quad\Longrightarrow\quad\int_\Sigma \tau\alpha=0\quad\forall\tau\in\mathcal{S},
    \end{align}
    for $\alpha\in\Omega_\Sigma^{2,1}$. Now, consider an $\alpha\in\Omega_\Sigma^{2,1}$ such that $p_{\mathcal T}\alpha=0$ holds together with the structural constraint $e_n(\alpha-p_{\mathcal T}\alpha)=e\sigma$ (notice that this subset of $\Omega_\Sigma^{2,1}$ is in general non-trivial because we do not require the condition $\alpha\in\mathrm{Ker}W_1^{\Sigma, (2,1)}$ as in \cref{strconstr_free_deg}), then it follows that
    \begin{align}
        \int_\Sigma \tau\alpha=\int_\Sigma ec\alpha+e_n\beta\alpha=\int_\Sigma ec p_{\mathcal{T}^C}\alpha+\beta e\sigma=\int_\Sigma ec p_{\mathcal{T}^C}\alpha,
    \end{align}
    where $p_{\mathcal{T}^C}$ is the projection onto a complement of $\mathcal T$. Since the right hand side of \eqref{implicatio_lemma_tau_1} must hold for all $\tau\in\mathcal S$, if the intersection $\mathcal S\cap \mathrm{Im}W_1^{\Sigma,(0,2)}$ were not trivial, we would have an absurdum. This implies $c\in\mathrm{Ker}W_1^{\Sigma,(0,2)}$ for all $\tau\in\mathcal S$, which, thanks to the injectivity of $W_1^{\Sigma,(0,2)}$, is equivalent to $c=0$.
    
    Lastly, the fact that $e_n\beta\in\mathrm{Ker}\tilde{\varrho}^{1,3}$ follows immediately from the definition of $\mathcal S$.
\end{proof}
\begin{proposition}\label{diagtetr}
    Let $\tau\in\mathcal{S}$ and $e$ be a \textit{diagonal degenerate} boundary vielbein, i.e. $e^*\eta=i^*\tilde{g}$ with $\eta=\mathrm{diag}(1,1,1-1)$ and $i^*\tilde{g}= \mathrm{diag}(1,1,0)$. Then, we have
    \begin{align}
        e_n[\tau,e]=0.
    \end{align}
\end{proposition}
\begin{proof}
    Given $a=1,2,3,4$ and let $\mu=1,2,+$ be the coordinates on the boundary $\Sigma$ such that we can write the diagonal degenerate boundary vielbein $e$ as
    \begin{align}
    \hat e^a&=
            \begin{cases}
                e_1^a&=\delta_1^a\\
                e_2^a&=\delta_2^a\\
            \end{cases}\\[5pt]
    e_+^a&=\delta_3^a-\delta_4^a\\[5pt]
    e_n^a&=\delta_3^a+\delta_4^a.
  \end{align}
Then, the definition of $\tau\in\mathcal S$ implies the following relations
\begin{align}
\tau_+^{abc}&=0\quad\forall a,b,c\\[3pt]
\tau_\mu^{123}&=0\quad\mu=1,2\\[3pt]
\tau_\mu^{124}&=0\quad\mu=1,2\\[3pt]
\tau_1^{234}&=\tau_2^{134}\\[3pt]
\tau_1^{134}&=-\tau_2^{234}.
\end{align}
The proof follows simply by computing $e_n[\tau,e]$ in components implementing the explicit form of the diagonal vielbein above\footnote{We refer to \cite{T2019b} for further details about this kind of computations.}.
\end{proof}
\begin{lemma}
    Let\footnote{Notice that this may be also a \textit{shifted} variable, like $\tau$ for example.} $A\in\Omega_\Sigma^{k,i}$ with $2\leq i\leq4$. Then, it holds
\begin{align}\label{identity_spinor_5}
\gamma\iota_\gamma\iota_\gamma A=(-1)^{|A|}(\iota_\gamma\iota_\gamma A\gamma+4(i-1)[\gamma,A]).
\end{align}
\end{lemma}
\begin{proof}
    \begin{align}
\gamma\iota_\gamma\iota_\gamma A&=(i-2)!\gamma^a\gamma^b\gamma^cv_a\iota_{v_b}\iota_{v_c}A\\
    &=-(i-2)!(\gamma^b\gamma^a\gamma^c+2\eta^{ab}\gamma^c)v_a\iota_{v_b}\iota_{v_c}A\\
    &=-(i-2)!(-\gamma^b\gamma^c\gamma^a+4\eta^{ab}\gamma^c)v_a\iota_{v_b}\iota_{v_c}A\\
    &=(-1)^{|A|}(\iota_\gamma\iota_\gamma A\gamma+4(i-1)[\gamma,A]).
\end{align}
\end{proof}
\begin{remark}
    This lemma introduces a relation between the action of the brackets over the Clifford algebra and $\mathcal V$-algebra. In particular, it is consistent a triviality condition on the bracket in the Clifford algebra, i.e.
    \begin{align}\label{identity_total_bracket}
    [A,\overline\psi\gamma\psi]=(-1)^{|A|}\overline\psi[A,\gamma]_{\mathcal V}\psi=\overline\psi\gamma[A,\psi]_{Cl}+ [\overline\psi,A]_{Cl}\gamma\psi,
    \end{align}
    where we occasionally added some redundancy with the labels of the specific algebras, even if we will not use them in general.
\end{remark}
\begin{lemma}\label{lemma_switch_bracket}
    Given $A\in\Omega_\Sigma^{k,i}$ and $B\in\Omega_\Sigma^{l,j}$ with $i,j=2,3$ such that $i+j<6$, then we have
    \begin{align}\label{deltatau_2}
        B(\overline\psi\gamma[A,\psi]- [A,\overline\psi]\gamma\psi)=(-1)^{|A||B|}A(\overline\psi\gamma[B,\psi]- [B,\overline\psi]\gamma\psi).
    \end{align}
\end{lemma}
\begin{proof}
The proof goes by direct computation of
\begin{align}
    B\gamma\iota_\gamma\iota_\gamma A&=B\gamma\gamma^a\gamma^b[v_a,[v_b,A]]\\[3pt]
    &=(-1)^{|B|}([v_a,B]\gamma+(-1)^{|B|}B\gamma_a)\gamma^a\gamma^b[v_b,A]\\[3pt]
    &=(-1)^{|B|}([v_a,B]\gamma\gamma^a-4(-1)^{|B|}B)\gamma^b[v_b,A]\\[3pt]
    &=-\big([v_b,[v_a,A]]\gamma\gamma^a\gamma^b-(-1)^{|B|}([v_a,B]\gamma_b\gamma^a\gamma^b+4[\gamma,B])\big)A\\[3pt]
    &=-(-(-1)^{|B|}\gamma\iota_\gamma\iota_\gamma B-6(-1)^{|B|}\iota_\gamma B)A\\[3pt]
    &=(-1)^{|B|}(-1)^{|A|(|B|+1)}A(\gamma\iota_\gamma\iota_\gamma B+6\iota_\gamma B)
\end{align}
and
\begin{align}
    \hspace{-0.12cm}B\iota_\gamma\iota_\gamma A\gamma&=(-1)^{|A||B|}\gamma^a\gamma^b[v_a,[v_b,A]]B\gamma\\[3pt]
    &=(-1)^{|A||B|}(-1)^{|A}\gamma^a\gamma^b[v_b,A]([v_a,B]\gamma+(-1)^{|B|}\gamma_aB)\\[3pt]
    &=-(-1)^{|A||B|}A\gamma^a\gamma^b\big([v_b,[v_a,B]]\gamma-(-1)^{|B|}([v_a,B]\gamma_b-\gamma_a[v_b,B])\big)\\[3pt]
    &=-(-1)^{|A||B|}A\big(-\iota_\gamma\iota_\gamma B\gamma+(-1)^{|B|}(4[\gamma,B]+\gamma^a\gamma^b\gamma_a[v_b,B])\big)\\[3pt]
    &=(-1)^{|A||B|}A(\iota_\gamma\iota_\gamma  B\gamma -(-1)^{|B|}6\iota_\gamma B).
\end{align}
Then, we can conclude the proof by considering the four possible parities of $A$ and $B$.
\end{proof}

\printbibliography
\end{document}